\documentclass[preprint,12pt,a4paper,twoside,fleqn,sort&compress]{elsarticle}
\usepackage{amssymb,latexsym,amscd,amsmath,amsthm}
\usepackage{pst-all}
\usepackage[varg]{pxfonts}
\usepackage{mathrsfs}
\usepackage{shuffle}
\usepackage{hyperref}
\usepackage{times}

\usepackage{color} %%%red, green, blue, cyan, magenta, and yellow

\usepackage{mathrsfs}
\usepackage{graphicx}

\usepackage{helvet}
\usepackage{courier}
\usepackage{booktabs,array,enumitem}
\usepackage{algorithm}
\usepackage[noend]{algorithmic}
\usepackage{verbatim}
\usepackage{flafter}
\usepackage{subcaption} % 需要加载该宏包
\theoremstyle{definition}
\newtheorem{definition}{Definition}%[section]
\theoremstyle{plain}
\newtheorem{proposition}{Proposition}%[section]
\newtheorem{theorem}{Theorem}%[section]
\newtheorem*{thm-other}{Theorem}
\newtheorem{example}{Example}%[section]
\newtheorem{corollary}{Corollary}
\newtheorem{lemma}{Lemma}
\theoremstyle{remark}
\newtheorem{remark}{Remark}

\newcommand{\Po}{\mathit{Po}}
\newcommand{\Ne}{\mathit{Ne}}
\newcommand{\Next}{\bigcirc}
\newcommand{\Until}{\mathbin{\sqcup}}
\newcommand{\Release}{\mathbin{R}}
\newcommand{\ecl}{\operatorname{ecl}}

\newcommand{\rank}{\operatorname{rank}}
\newcommand{\tp}{\operatorname{tp}}
\newcommand{\DAG}{\operatorname{DAG}}
\newcommand{\DAGG}{\operatorname{DAGG}}
\newcommand{\Paths}{\operatorname{Paths}}
\newcommand{\Val}{\operatorname{Val}}
\newcommand{\FVPKS}{\mathrm{FVPKS}}
\newcommand{\AxSys}{\mathrm{AxSys}_{\mathrm{PoCTL}}}
\newcommand{\true}{\mathit{true}}
\newcommand{\false}{\mathit{false}}
\newcommand{\bsim}{\overline{\sim}}

\begin{document}

%\journal{Theoretical Computer Science}
\journal{arXiv}

\begin{frontmatter}

%\title{\Large\bf
%\thanks{This work is supported by National Science Foundation of China (Grant No: 60873119)  and the Higher School Doctoral Subject Foundation of Ministry of Education of China (Grant No:200807180005).}}
%\author{{Yongming Li, Qian Wang, Sanjiang Li}\\
%  {\small College of Computer Science,}
%  {\small Shaanxi Normal University, Xi'an, 710062, China}\\
%  {\small Email:liyongm@snnu.edu.cn, Tel.: +862985310166, Fax:  +862985310161}}

\title{Possibilistic Computation Tree Logic over\\
Finitely-Valued Kripke Structures:\\
Decidability and Complete Axiomatization
\thanks{This work was partially supported by National Science
Foundation of China (Grant No: 12471437).}}
\date{}
\author {Yongming Li\corref{cor1}}
\ead{liyongm@snnu.edu.cn}

\address {School of Mathematics and Statistics, Shaanxi Normal University, Xi'an, 710119, China}
%\address {Department of Foundational Teaching, Beifang University of Nationalities, Yinchuan 750021}
\cortext[cor1]{Corresponding Author}
\begin{abstract}
Possibilistic computation tree logic (PoCTL) is a branching-time
temporal logic for specifying and verifying systems whose behavior is
described by possibility theory. Although the model-checking problem
for PoCTL has been investigated,  its satisfiability problem and proof-theoretic
foundations, particularly
the development of a sound and complete axiomatization,
remain largely unexplored. We investigate these
problems over normalized possibilistic Kripke structures whose transition
possibilities take finitely many values; their state spaces may be countably
infinite. A counterexample shows that unrestricted countable structures
do not have the finite-model property, thereby distinguishing the semantic
scope of the present results. For the finitely-valued class, we extract a
finite threshold scale from the input formula and define canonical
admissible transitions by the set $D(s,t)$. Independent eventuality ranks
guide the selection of local witnesses. Finite supported fragments are then
spliced cyclically to obtain a model, without requiring the ranks of
different eventualities to decrease simultaneously. This yields a finite-model
property. Satisfiability for the unbounded language over this class is
EXPTIME-complete, and the tableau procedure runs in time $2^{O(n^2)}$.
Bounded operators are handled by
explicit finite unfolding, with complexity measured after this preprocessing.
Finally, we establish a sound and weakly complete finitary axiomatization
by translating tableau deletion into formal refutations.

\end{abstract}


\end{frontmatter}

\baselineskip 20pt

\section{Introduction}

{ Started} with the work of { Arthur} Prior \cite{Prior1957}, temporal logics serve to analyse and reason about system behaviors changing with time. As an extension of classical logic, it incorporates temporal operators that enable the expression of propositions about the order of events and temporal relationships. Temporal logics are useful for the specification and verification of computation systems, such as hardware and software systems (\cite{Pnueli77,Clarke86,EGP92,Baier08}). %Classical temporal logics include linear temporal logic (LTL) \cite{LTL} and computation tree logic (CTL) \cite{CTL}. LTL is called linear because its qualitative concept of time is path-based and considered linear: at each moment in time, there is only one possible successor state, and thus each time moment has a unique possible future \cite{principles-baier2008}. LTL focuses on properties of future states, and its temporal logic operators include the next operator ``$\bigcirc$'', the until operator ``$\sqcup$'', and derived operators ``$\square$'' (``always'', from now on forever) and ``$\lozenge$'' (``eventually'', sometimes in the future). These operators enable the formalization of various temporal relationships. Using LTL formulas, we can formalize temporal properties such as ``an event will occur at some point in the future'' or ``a condition will always hold''. In LTL, system behavior is assumed to be deterministic. This means that each state transition of a system is fixed, and its future behavior is fully predictable.
However, the temporal behavior of many real-world systems is not entirely predictable and may be influenced by uncertainty such as probability and possibility. This phenomenon { occurs} often in the diagnosis of intelligent systems and expert systems (c.f. \cite{Almagor16,Bouyer23}). A typical example is the diagnosis of a patient, in which the states of the patient and the changes of the states of the patient under the treatment are not certainly known \cite{li13}, it is dependent on the data or the observation of the doctors and medial instruments.

In this work, we focus on the branching temporal logic in terms of possibilistic logic. We have introduced the possibilistic computational tree logic (PoCTL) in \cite{PoCTL2015,Xue11}. In fact,
PoCTL is a possibilistic extension of the well-known branching-time logic obtained by replacing the existence and universal path quantifiers with the possibilistic operators, which allow to quantify the possibility of all runs that satisfy a given path formula. Formally, PoCTL formulae are interpreted over possibilistic Kripke structures (PKSs) where each state is assigned a subset of atomic propositions that are valid in a given state.

PoCTL is a possibility measure extension of classical CTL. Both the possibilistic and probabilistic CTL solve certain uncertainty of errors or other stochastic behaviors occurring in various real-world applications. As shown in \cite{Baier08}, probabilistic CTL and CTL are not comparable with respect to their expressiveness, while CTL { is} the proper sub-logic of PoCTL with respect to their expressiveness as shown in \cite{PoCTL2015}. This allows possibilistic CTL to be used to represent more { instances} of real-word systems compared with the classical and probabilistic CTL. Using this advantage of PoCTL, we can use PoCTL to solve the model-checking problems of real-world applications, which can not be tackled by classical model-checking algorithms \cite{PoCTL2015}.

As a logic system, there are three logic decision problems for PoCTL, the satisfiability problem, the validity problem and the model checking problem. The model checking problem of PoCTL had been studied in \cite{PoCTL2015} in detail. Whereas the satisfiability and validity problems of PoCTL are not solved. Since the satisfiability and validity problems are dual to each other, we only consider the satisfiability problem of PoCTL, which asks for a model for a given PoCTL formula.

Axiomatization is another important aspect of logic system. Up to now, there is few work on the axiomatization of possibilistic temporal logic, especially the axiomatization of PoCTL is not studied. We shall discuss the axiomatization of PoCTL in this paper.

For the related work, the satisfiability and axiomatization of classical temporal logics was studied extensively, we refer to \cite{EH85,Emerson90,Demri16} for the details and references therein.
The satisfiability and axiomatization of probabilistic temporal logics had been studied recently. Ognjanovi\'{c} introduced discrete linear-time probabilistic logics and provided a sound and complete infinitary axiomatic system \cite{PLTL2006}. Their study focuses on finitely additive probability measures. Doder et al. \cite{Doder2024} proposed a temporal logic, called $PL_{LTL}$, based on polynomial weight formulas (PWFs) and LTL. $PL_{LTL}$ enables reasoning about probabilistic events. They provided $\sigma$-additive semantics for $PL_{LTL}$ and proved that its axiomatic system is sound and strongly complete. Along with this direction, He et al. studied the complete axiomatization of posibilistic linear temporal logic (PoLTL). For the axiomatization of probabilistic branching-time logic, in \cite{PCTL*-axiomatization}, Ognjanovi\'{c} et al. introduced two probabilistic operators into CTL$^*$ and provided a sound and strongly complete infinitary axiomatic system.

Recent results on PCTL emphasize how difficult quantitative satisfiability can be: the satisfiability and validity problems of full PCTL are highly undecidable~\cite{Chodil24,Chodil25}.  One reason is that the quantitative information needed to build a model cannot in general be reduced to a small finite set of probability constraints.  For PoCTL, the analogous challenge is to extract enough possibility information from a formula to control the transition weights of a prospective model.  The maxitive structure of possibility measures makes such an extraction feasible, but the precise semantic domain matters.

For general countable possibilistic Kripke structures, a non-strict lower
bound may hold only as an unattained supremum. Similarly, the infimum
along an infinite path may equal a threshold even when every transition
on that path is strictly above it. These phenomena obstruct a direct
finite-witness construction. Section~\ref{sec:counterexample} gives a
formula that is satisfiable in a normalized countable structure with
infinitely many transition values, but has no finitely-valued model.
Consequently, unrestricted PoCTL does not have the finite-model property.

We therefore study PoCTL over \emph{finitely-valued} possibilistic Kripke
structures. The restriction concerns the set of transition values
throughout the structure; the state set need not be finite. At each state,
path possibilities and nonempty path-event possibilities are then attained.
This yields exact threshold witnesses while retaining countably infinite
systems. No value scale is fixed in advance: the scale used in a finite
model is extracted from the formula.

Our construction follows the tableau and finite-fragment approach to
branching-time logic~\cite{EH85,Emerson90}. The set $D(s,t)$ records the weights
compatible with the necessity-next formulas at a tableau node. Its maximum
is an admissible upper bound for an edge, rather than a requirement to
retain every such edge. Possibility eventualities use the usual min rank.
Necessity eventualities use a max--min rank over the relevant local
possibility-next requirements, each evaluated at its own threshold.
These independent ranks provide finite supported fragments. Cyclic
splicing ensures that every eventuality receives an opportunity for
completion. In particular, local witnesses for different eventualities
need not decrease all their ranks simultaneously.

The main results are a finite-model property, an EXPTIME-complete
satisfiability problem for the unbounded language, and a sound and weakly
complete finitary axiomatization over the finitely-valued class.
The tableau runs in time $2^{O(n^2)}$ under the size convention specified
below. The axiomatization retains the standard Boolean, threshold,
unfolding, and induction schemes. Its proof uses complete-type
decompositions and refutes every deleted tableau node. We also obtain a
tree-model property and show why strong completeness for arbitrary
premise sets cannot be achieved by a finitary system.

The content of this paper is arranged as follows. In Section 2 we recall the notion of possibilistic Kripke structures, the related possibility measures induced by the possibilistic Kripke structures, and the main notions of PoCTL and the positive normal form of PoCTL. In Section 3, the satisfiability problem and the small model theorem of PoCTL are studied.
An axiomatization system of PoCTL is presented in Section 4.
The paper ends with conclusion section.

\section {Preliminaries}\label{sec:prelim}
\subsection{Possibility theory}
Possibility theory, introduced by Zadeh in 1978 as an extension of fuzzy sets and fuzzy logic \cite{Zadeh78}, provides a framework for dealing with vagueness and imprecision in systems. Unlike probability theory, possibility theory employs a pair of dual set functions (possibility and necessity measures) instead of only one. Further studies were conducted by Dubois et al. in \cite{Dubois88,Dubois01}.

 For simplicity, assume that the universe of discourse $U$ is a
 nonempty set, and assume that all subsets are measurable. A
 {possibility measure is a function $Po$ from the powerset $2^U$ to
 $[0, 1]$ such that:

 1) $Po(\emptyset)=0$, 2)$Po(U)=1$, and 3) $Po(\bigcup_{i\in I} E_i)={\rm sup}_{i\in I} Po(E_i)$ for
 any subset family $\{E_i\}$ of the universe set $U$, where we use
 ${\rm sup}_{i\in I}a_i$ to denote the supremum or the least upper bound of
 the family of real numbers $\{a_i\}_{i\in I}$; dually, we use ${\rm inf}_{i\in I}a_i$ to
 denote the infimum or the largest lower bound of the family of
 real numbers $\{a_i\}_{i\in I}$.

 Whereas probability theory uses a single number, the probability, to describe how likely an event is to occur, possibility
 theory uses two concepts: the possibility and the necessity of
 the event (\cite{Dubois88,grabisch00}).

 For any set $E$, the necessity measure $Ne$ is
 defined by
\begin{equation}
  Ne(E)=1-Po(U-E).
\end{equation}

 A necessity measure is a function $Ne$ from the powerset $2^U$ to
 $[0, 1]$ such that:

 1) $Ne(\emptyset)=0$, 2)$Ne(U)=1$, and 3)$Ne(\bigcap_{i\in I} E_i)= {\rm inf}_{i\in I} Ne(E_i)$ for
 any subset family $\{E_i\}_{i\in I}$ of the universe set $U$.

  It follows that $Po(E)+Ne(U-E)=1$, and $Ne$ is the dual of
 $Po$ and vice versa. In general, $Po$ and $Ne$ are not self-dual, i.e., $Po\not=Ne$; this is
 contrary to probability measure, which is self-dual. As a result,
 we need both possibility measure and necessity measure to treat
 uncertainty in the theory of possibility.

  We shall use possibility measures and necessity measures in
 the logic PoCTL in this paper.

\subsection{PoCTL formulae and their positive normal  forms}

\subsubsection{Syntax of PoCTL}

\begin{definition} \cite{PoCTL2015} (Syntax of PoCTL) {\sl PoCTL state formulae} over the set $AP$ of atomic propositions are formed according to the following grammar:
\begin{center}
$\Phi ::= true\mid a \mid\Phi_{1} \wedge \Phi_{2}\mid \neg \Phi\mid Po_{J}(\varphi)$
\end{center}
where $a\in AP$, $\varphi$ is a PoCTL path formula and $J\subseteq [0,1]$ is an interval with rational bounds.

{\sl PoCTL path formulae} are formed according to the following grammar:

\begin{center}
$\varphi::=\bigcirc \Phi \mid \Phi_{1}\sqcup \Phi_{2}\mid \Phi_{1}\sqcup^{\leq n} \Phi_{2}\mid \square\Phi $
\end{center}
where $ \Phi$, $ \Phi_{1}$, and $ \Phi_{2}$ are state formulae and $n\in\mathbb{N}$.
\end{definition}

Let $\Phi_{1}$ and $\Phi_{2}$ be PoCTL formulae. Using the Boolean connectives $\wedge$ and $\neg$, the full power of propositional logic is obtained. Other Boolean connectives, such as disjunction $\vee$, implication $\rightarrow$, and equivalence $\leftrightarrow$, can be derived as follows:
\begin{equation*}
\begin{split}
   \Phi_{1}\vee\Phi_{2} & \overset{\text{def}}=\neg(\neg\Phi_{1}\wedge \neg\Phi_{2}), \\
    \Phi_{1}\rightarrow \Phi_{2} &\overset{\text{def}}=\neg\Phi_{1}\vee \Phi_{2}, \\
    \Phi_{1}\leftrightarrow \Phi_{2} & \overset{\text{def}}=(\Phi_{1}\rightarrow\Phi_{2})\wedge(\Phi_{2}\rightarrow \Phi_{1}).
\end{split}
\end{equation*}

\subsubsection{Semantics of PoCTL}

Possibilistic Kripke structures are key models of possibilistic temporal logic, which is defined as follows.

\begin{definition} \cite{li13}\label{de:pkripke}
A possibilistic Kripke structure (PKS, in short) is a tuple $M=(S,P,I,AP,L)$, where

(1) $S$  is a countable, nonempty set of states;

(2) $P:S\times S\longrightarrow [0,1]$ is the transition possibility distribution such that for all states $s$, {\sl the normal condition} ${\rm sup}_{s^{'}\in S}P(s,s^{'})=1$ holds;

(3) $I:S\longrightarrow[0,1]$ is the initial distribution, such that the normal condition ${\rm sup}_{s\in S}I(s)=1$ holds;

(4) $AP$ is a set of atomic propositions;

(5) $L:S\longrightarrow 2^{AP}$  is a labeling function that labels a state $s$ with those atomic propositions in $AP$ that are supposed to hold in $s$.

Furthermore, if the set  $S$  and  $AP$ are finite sets, then $M=(S,P,I,AP,L)$ is called a
finite possibilistic Kripke structure.
\end{definition}

Paths in a PKS $M$ are infinite paths in the underlying digraph. They are
defined as infinite state sequences $\pi=s_{0}s_{1}s_{2}\cdots\in S^{w}$  such that  $P(s_{i},s_{i+1})>0$ for all $i\in I$.
Let $Paths(M)$ denote the set of all paths in $M$. Let $Paths_M(s)$ ($Paths(s)$ if $M$ is understood) denote the set of all
paths in $M$ that start in state $s$.

\begin{definition}[Finitely-valued PKS]\label{def:fv}
A PKS $M$ is finitely valued if
\[
 \Val(M)=\{P(s,t):s,t\in S\}
\]
is finite. The class of normalized finitely-valued PKSs is denoted
$\FVPKS$. The state set may still be countably infinite.
\end{definition}

This is a global finite-value condition. Requiring each individual state
to have only finitely many outgoing values would be weaker.
The initial distribution is not part of $\Val(M)$, since the logical
results below use state-local semantics.

\begin{definition}\label{def:possibility measure} \cite{li13} For a PKS $M$, a function $Po^M: Paths(M)\rightarrow [0,1]$ is defined as follows:
\begin{equation}\label{eq:possibility measure-path}
Po^M(\pi)={\rm inf}(\{I(s_{0})\}\cup\{ P(s_{i},s_{i+1}): i\geq 0\})
\end{equation}
for any $\pi=s_{0}s_{1} \cdots, \pi\in Paths(M).$
Furthermore, we define
\begin{equation}\label{eq:possibility measure}
Po^M(E)={\rm sup}\{Po^M(\pi): \pi\in E\}
\end{equation}
for any $E\subseteq Paths(M)$, then, we have a well-defined function $$Po^M:2^{Paths(M)}\longrightarrow [0,1],$$ $Po^M$ is called the possibility measure over $\Omega=2^{Paths(M)}$ as it has the properties stated in Theorem \ref{th:possibility measure}. If $M$ is clear from the context, then $M$ is omitted and we simply write $Po$ for $Po^M$.

\end{definition}

\begin{theorem}\label{th:possibility measure} \cite{li13} $Po$ is a possibility measure on $\Omega=2^{Paths(M)}$, i.e., $Po$ satisfies the following conditions:

 (1) $Po(\varnothing)=0$, $Po(Paths(M))=1$;

 (2) $Po(\bigcup\limits_{i\in I}A_{i})={\rm sup}\{Po(A_{i}): i\in I\}$ for any $A_{i}\in\Omega$, $i\in I$.

\end{theorem}

For evaluation at a state $s$, make $s$ the unique initial state and
denote the resulting PKS by $M_s$. For $\pi\in\Paths_M(s)$ this gives
\[
 \Po_{M_s}(\pi)=\inf_{j\ge0}P(\pi[j],\pi[j+1]).
\]
All attainment and threshold arguments below refer to these state-local
path measures. When $AP$ and the initial distribution are understood,
write $M=(S,P,L)$ and $\Paths(s)=\Paths_M(s)$.

\begin{definition}\label{def:semantics of PoCTL}\cite{PoCTL2015} (Semantics of PoCTL) Let $a\in AP$ be an atomic proposition, $M=(S,P,L)$ be a possibilistic Kripke structure, state $s\in S$, $\Phi$, $\Psi$ be PoCTL state formulae, and $\varphi$ be a PoCTL path formula. {\sl The satisfaction relation $\models$} is defined {\sl for state formulae} by
\begin{eqnarray*}
M, s\models a  & {\rm iff} \ a\in L(s);\\
M, s\models\neg\Phi & {\rm iff}\ s\not\models\Phi; \\
M, s\models\Phi\wedge\Psi  & {\rm iff} \ s\models\Phi \ {\rm and}\ s\models\Psi;\\
M, s\models Po_{J}(\varphi)  & {\rm iff} \ Po(s\models \varphi)\in J, \ {\rm where}\ Po(s\models \varphi)=Po^{M_s}(\{\pi | \pi\in Paths(s), \pi\models \varphi\}).
\end{eqnarray*}

For path $\pi$, {\sl the satisfaction relation $\models$ for path formulae} is defined by
\begin{eqnarray*}
\pi\models\bigcirc\Phi  & {\rm iff} \ \pi[1]\models\Phi;\\
\pi\models\Phi\sqcup\Psi  & {\rm iff} \  \exists k\geq0,\pi[k]\models\Psi
\ {\rm and}\
 \pi[i]\models\Phi {\rm \ for\ all} \ 0\leq i\leq k-1;\\
\pi \models \Phi\sqcup^{\leq n}\Psi & {\rm iff} \ \exists 0\leq k\leq n, (\pi[k]\models \Psi\wedge(\forall 0\leq i< k),\pi[i]\models \Phi));\\
\pi \models\square\Phi & {\rm \ iff}\ \pi[j]\models \Phi {\rm \ for\ all\ } j\geq 0.
\end{eqnarray*}
where if $\pi=s_0s_1s_2\cdots$, then $\pi[k]=s_k$ for any $k\geq 0$.
\end{definition}

In particular, the path formulae $\lozenge\Phi$ (``eventually'') has the semantics
$$\pi=s_0s_1\cdots\models \lozenge\Phi {\rm \ iff}\ s_j\models \Phi {\rm \ for\ some\ } j\geq 0,$$
Alternatively, $\lozenge\Phi=true\sqcup \Phi$.

We can define release operator $\Release$ as follows, for a path $\pi$, $\pi\models \Phi_{1}\Release \Phi_{2}$ iff $\forall j (\pi[j]\not\models \Phi_{2}$ implies $\exists k<j (\pi[k]\models \Phi_{1}))$. Then it follows that $\square \Phi={ false}\  \Release \Phi$.

If we write a PoCTL formula $Ne_J(\varphi)$ for a path formula $\varphi$, which have the semantics
\begin{equation}
s\models Ne_J(\varphi) \ {\rm iff}\ Ne^{\mathcal{M}_s}(\{\pi\in Paths(s) : \pi\models \varphi\})\in J
\end{equation}
\noindent for any PKS $M$, then we have the following representation of $Ne_J(\varphi)$, where for interval $J=[u,v], (u,v], [u,v), (u,v)$, $DJ=[1-v,1-u], [1-v,1-u), (1-v, 1-u], (1-v,1-u)$:

\begin{eqnarray}
Ne_J(\bigcirc \Phi)&=& Po_{DJ}(\bigcirc \neg \Phi),\\
Ne_J(\Phi\sqcup \Psi)&=& %Po_{DJ}(\neg\Psi\sqcup(\neg\Phi\wedge \neg\Psi))\wedge Po_{DJ}(\square \neg\Psi)=
Po_{DJ}(\neg\Phi \Release\neg\Psi),\\
Ne_J(\Phi\sqcup^{\leq n} \Psi)&=& { Po_{DJ}(\neg\Phi\Release^{\le n}\neg\Psi),}\\
Ne_J(\square\Phi)&=& Po_{DJ}(\lozenge\neg\Phi),
\end{eqnarray}
where, for bounded until, the bounded release as a definitional
path abbreviation:
\[
 \Phi\Release^{\le n}\Psi:=\neg(\neg\Phi\Until^{\le n}\neg\Psi).
\]

\begin{definition}[Satisfiability, validity, and consequence]\label{def:sat}
A formula $\Phi$ is satisfiable over $\FVPKS$ if $M,s\models\Phi$
for some $M\in\FVPKS$ and some state $s$.
Write $\models_{\mathrm{FV}}\Phi$ when $\Phi$ holds at every state of
every member of $\FVPKS$.
For a set $\Gamma$ of formulas, write $\Gamma\models_{\mathrm{FV}}\Phi$
if every such state satisfying all members of $\Gamma$ also satisfies $\Phi$.
Equivalence is denoted by $\equiv_{\mathrm{FV}}$.
\end{definition}
From now on, unqualified satisfiability, validity, equivalence, and consequence
in the main results refer to $\FVPKS$.
Statements about unrestricted countable PKSs are identified explicitly.
As usual, $\models_{\mathrm{FV}}\Phi$ iff $\neg\Phi$ is not satisfiable
over $\FVPKS$.

Before applying temporal unfoldings, simplify the degenerate bounds:
\begin{equation}\label{eq:endpoints}
 \Po_{\ge0}(\varphi)\equiv\Ne_{\ge0}(\varphi)\equiv\true,
 \qquad
 \Po_{>1}(\varphi)\equiv\Ne_{>1}(\varphi)\equiv\false.
\end{equation}
A bound $\ge r$ is called nondegenerate when $0<r\le1$,
and a bound $>r$ when $0\le r<1$.

\begin{lemma}[Finite-value attainment and exact thresholds]\label{lem:fv}
Let $M\in\FVPKS$ and $s\in S$. Every state-local path possibility is
a minimum of finitely many transition values. Every nonempty event
in $\Paths(s)$ has a path of maximum possibility.
Moreover, some $t$ satisfies $P(s,t)=1$.
For nondegenerate bounds, the following equivalences hold:
\begin{align}
 M,s\models\Po_{\ge r}(\varphi)
 &\ \Longleftrightarrow\
 \exists\pi\in\Paths(s)\
 [\pi\models\varphi\ \text{and}\ \forall j\ P(\pi[j],\pi[j+1])\ge r],
 \label{eq:cut-po-ge}\\
 M,s\models\Po_{>r}(\varphi)
 &\ \Longleftrightarrow\
 \exists\pi\in\Paths(s)\
 [\pi\models\varphi\ \text{and}\ \forall j\ P(\pi[j],\pi[j+1])>r],
 \label{eq:cut-po-gt}\\
 M,s\models\Ne_{>r}(\varphi)
 &\ \Longleftrightarrow\
 \forall\pi\in\Paths(s)\
 [\,\forall j\ P(\pi[j],\pi[j+1])\ge1-r\ \Rightarrow\ \pi\models\varphi\,],
 \label{eq:cut-ne-gt}\\
 M,s\models\Ne_{\ge r}(\varphi)
 &\ \Longleftrightarrow\
 \forall\pi\in\Paths(s)\
 [\,\forall j\ P(\pi[j],\pi[j+1])>1-r\ \Rightarrow\ \pi\models\varphi\,].
 \label{eq:cut-ne-ge}
\end{align}
\end{lemma}
\begin{proof}
The transition values on any path belong to the finite set $\Val(M)$,
so their infimum is their minimum. Thus path possibilities themselves
take only finitely many values. The supremum over any nonempty path
event is therefore attained. Normalization is attained for the same
reason, applied to the outgoing values at $s$.

The two possibility clauses now follow directly. For the strict case,
if every edge of a path has weight greater than $r$, then
\[
 \inf_{j\ge0}P(\pi[j],\pi[j+1])
 \ge \min\{v\in\Val(M):v>r\}>r.
\]
The set on the right is nonempty because the path exists, and its
minimum exists because $\Val(M)$ is finite.
No assumption that $r\in\Val(M)$ is needed.
For necessity, use $\Ne(A)=1-\Po(A^c)$.
The condition $\Ne(A)>r$ excludes a counterexample path of possibility
at least $1-r$, and $\Ne(A)\ge r$ excludes one of possibility greater
than $1-r$. Attainment and the pathwise minimum give the last two clauses.
\end{proof}

Every finite prefix can consequently be extended along weight-$1$ edges.
In particular, a finite until witness retains its threshold after this
extension. We refer to edges satisfying the respective conditions in
\eqref{eq:cut-po-ge}--\eqref{eq:cut-ne-ge} as \emph{relevant edges}
for the formula in question.
Every state has a relevant successor for each nondegenerate bound,
since a weight-$1$ edge is always relevant.
In particular, the next-step instances used below are
\begin{align}
 M,s\models\Po_{\sim r}(\Next\Phi)
 &\ \Longleftrightarrow\
 \exists t\ [P(s,t)\sim r\ \text{and}\ M,t\models\Phi],
 \label{eq:po-sim}\\
 M,s\models\Ne_{>r}(\Next\Phi)
 &\ \Longleftrightarrow\
 \forall t\ [P(s,t)\ge1-r\ \Rightarrow\ M,t\models\Phi],
 \label{eq:ne->}\\
 M,s\models\Ne_{\ge r}(\Next\Phi)
 &\ \Longleftrightarrow\
 \forall t\ [P(s,t)>1-r\ \Rightarrow\ M,t\models\Phi].\label{eq:ne-geq}
\end{align}
Each edge can be tested by extending its one-edge prefix along
weight-$1$ edges, after which Lemma~\ref{lem:fv} applies.

\subsubsection{The positive normal  form of PoCTL formulae}

For the { convenience} to give the model of a PoCTL formula, we need some normal forms of PoCTL formulae. For this purpose, we give some propositions about the equivalence of PoCTL formulae, which can be verified directly.

\begin{proposition}\label{pro: the equivalence of PoCTL formulae}
We have the following equivalence of PoCTL formulae.

(1)$Po_{[u,v]}(\varphi)=Po_{\geq u}(\varphi)\wedge Po_{\leq v}(\varphi)$,$Po_{(u,v]}(\varphi)=Po_{> u}(\varphi)\wedge Po_{\leq v}(\varphi)$,$Po_{[u,v)}(\varphi)=Po_{\geq u}(\varphi)\wedge Po_{< v}(\varphi)$,$Po_{(u,v)}(\varphi)=Po_{> u}(\varphi)\wedge Po_{<v}(\varphi)$ for any path formula $\varphi$.

(2)$Po_{<r}(\varphi)=\neg Po_{\geq r}(\varphi)$, $Po_{\leq r}(\varphi)=\neg Po_{>r}(\varphi)$, $Ne_{<r}(\varphi)=\neg Ne_{\geq r}(\varphi)$, $Ne_{\leq r}(\varphi)=\neg Ne_{>r}(\varphi)$ for any path formula $\varphi$.

(3)$\neg Po_{\geq r}(\bigcirc\Phi)=Ne_{>1-r}(\bigcirc \neg \Phi)$, $\neg Po_{\geq r}(\Phi\sqcup \Psi)=Ne_{>1-r}(\neg\Phi \Release\neg\Psi)$, $\neg Po_{\geq r}$ $(\Phi\Release\Psi)=Ne_{>1-r}(\neg\Phi \sqcup \neg\Psi)$, $\neg Po_{\geq r}(\lozenge\Phi)=Ne_{>1-r}(\square\neg\Phi)$ for any PoCTL formulae $\Phi$ and $\Psi$.

(4) $Po_{>r}(\varphi)\wedge Po_{\geq r}(\varphi)=Po_{>r}(\varphi)$, $Po_{>r_1}(\varphi)\wedge Po_{> r_2}(\varphi)=Po_{>r_1}(\varphi)$ and $Po_{\geq r_1}(\varphi)\wedge Po_{\geq r_2}(\varphi)=Po_{\geq r_1}(\varphi)$ whenever $r_1\geq r_2$, $Po_{>r_1}(\varphi)\wedge Po_{\geq r_2}(\varphi)=Po_{\geq r_2}(\varphi)$ whenever $r_2> r_1$.

(5) $Ne_{>r}(\varphi)\wedge Ne_{\geq r}(\varphi)=Ne_{>r}(\varphi)$, $Ne_{>r_1}(\varphi)\wedge Ne_{> r_2}(\varphi)=Ne_{>r_1}(\varphi)$ and $Ne_{\geq r_1}(\varphi)\wedge Ne_{\geq r_2}(\varphi)=Ne_{\geq r_1}(\varphi)$ whenever $r_1\geq r_2$, $Ne_{>r_1}(\varphi)\wedge Ne_{\geq r_2}(\varphi)=Ne_{\geq r_2}(\varphi)$ whenever $r_2> r_1$.
\end{proposition}

After \eqref{eq:endpoints}, bounded until is eliminated over $\FVPKS$ by
\begin{align}
 \Po_{\sim r}(\Phi\Until^{\le0}\Psi)&\equiv\Psi,\qquad
 \Ne_{\sim r}(\Phi\Until^{\le0}\Psi)\equiv\Psi,\notag\\
 \Po_{\sim r}(\Phi\Until^{\le n+1}\Psi)
 &\equiv\Psi\vee
 \bigl(\Phi\wedge
 \Po_{\sim r}(\Next\Po_{\sim r}(\Phi\Until^{\le n}\Psi))\bigr),
 \label{eq:bounded-po}\\
 \Ne_{\sim r}(\Phi\Until^{\le n+1}\Psi)
 &\equiv\Psi\vee
 \bigl(\Phi\wedge
 \Ne_{\sim r}(\Next\Ne_{\sim r}(\Phi\Until^{\le n}\Psi))\bigr).
 \label{eq:bounded-ne}
 \end{align}
To verify these identities, induct on $n$. At $n=0$ the path event is
either the whole path space or the empty event, according to $\Psi$.
At $n+1$, if $\Psi$ does not yet hold, $\Phi$ must hold and the suffix
must satisfy the $n$-step condition. Lemma~\ref{lem:fv} converts this
condition into exactly the displayed next modality.
Bounded release is eliminated by duality.

Using the above proposition, we can give the positive normal  form of PoCTL formulae as below.

\begin{definition} (Syntax of positive normal  form of PoCTL) {\sl PoCTL state formulae} over the set $AP$ of atomic propositions are formed according to the following grammar:
\begin{center}
$\Phi ::= true\mid p \mid\neg { p} \mid\Phi_{1} \wedge \Phi_{2}\mid\Phi_{1} \vee \Phi_{2}\mid Po_{\sim r}(\varphi)\mid Ne_{\sim r}(\varphi)$
\end{center}
where $p\in AP$, $\sim\in\{\geq, >\}$, $r$ is a rational number in $[0,1]$,  $\varphi$ is a PoCTL path formula.

{\sl PoCTL path formulae} are formed according to the following grammar:

\begin{center}
$\varphi::=\bigcirc \Phi \mid \Phi_{1}\sqcup \Phi_{2}\mid \Phi_{1}\Release \Phi_{2}\mid \lozenge \Phi \mid \square\Phi$
\end{center}
where $ \Phi$, $ \Phi_{1}$, and $ \Phi_{2}$ are state formulae.
\end{definition}

 %and $Ne_{>r}(\Phi_1\sqcup \Phi_2)=\neg Po_{\geq 1-r}(\neg \Phi_1\Release \neg\Phi_2)$ and $Ne_{>r}(\Phi_{1}\Release \Phi_{2})=\neg Po_{\geq 1-r}(\neg \Phi_1\sqcup \neg\Phi_2)$.

\begin{theorem}[Positive normal form]\label{thm:pnf}
Every PoCTL formula has an equivalent unbounded positive normal form
over $\FVPKS$ after finite bounded-operator preprocessing.
For the unbounded language, the transformation has polynomial size
when repeated subformulas are shared.
\end{theorem}
\begin{proof}
Split interval constraints, simplify degenerate endpoints, and use the
finite unfoldings above for bounded operators. Push negations inward
using De Morgan laws and the temporal and quantitative dualities.
For shared representations, process each subformula in each polarity
once. Interval splitting creates a bounded number of nodes referring to
the same path subformula, and negation does not require repeatedly
copying that subformula. This gives a polynomial representation for the
unbounded language.
\end{proof}

Throughout Section~\ref{sec:sat}, $n=|\Lambda|$ is the size of the resulting
unbounded positive normal formula $\Lambda$, counting its logical nodes and the
binary encoding of rational thresholds, with shared subformulas.
For bounded input formulas the bounds stated below are measured after
the explicit finite unfolding. In particular, they are not asserted
as bounds in the original bit length of a binary-encoded horizon.

\subsection{Why the finitely-valued restriction matters}\label{sec:counterexample}

General countable PKSs may have infinitely many transition values.
A non-strict threshold can then be realized only as a supremum.
The following example separates their satisfiability semantics from
the finitely-valued class and disproves a finite-model property for
unrestricted PoCTL.

\begin{example}[A satisfiable formula with no finitely-valued model]\label{ex:no-fmp}
Let
\begin{equation}\label{eq:counter-formula}
 \chi=\neg a\wedge\Ne_{>1/2}(\Next\neg a),
 \qquad
 \Lambda_\infty=\chi\wedge\Po_{\ge1/2}(\chi\Until a).
\end{equation}
Consider the countable structure $M_\infty=(S,P,L)$ with
\[
 S=\{s_n:n\ge1\}\cup\{z\},
  P(s_n,s_{n+1})=1,
 P(s_n,z)=q_n:=\frac12-\frac1{n+2}, P(z,z)=1.
\]
All other weights are $0$, and $a$ holds only at $z$.
The structure is normalized, while $0<q_n<1/2$ and $q_n\uparrow1/2$,
so it is not finitely valued.

The possibility of $\Next a$ at $s_n$ is exactly $q_n$.
Consequently,
\[
 \Ne_{M_{s_n}}(\Next\neg a)=1-q_n>\frac12,
 \qquad M_\infty,s_n\models\chi.
\]
For every $n$, the path $s_1s_2\cdots s_nzz\cdots$ satisfies
$\chi\Until a$ and has possibility $q_n$.
These are precisely the paths from $s_1$ that eventually reach $a$.
It follows that
\[
 \Po_{M_{s_1}}(\chi\Until a)=\sup_{n\ge1}q_n=\frac12,
 \qquad M_\infty,s_1\models\Lambda_\infty.
\]

Suppose, however, that some $M\in\FVPKS$ and state $u$ satisfy
$\Lambda_\infty$. Lemma~\ref{lem:fv} gives an actual path
$u=u_0u_1\cdots$ and $k\ge1$ such that
$M,u_k\models a$, $M,u_j\models\chi$ for $j<k$, and
$P(u_j,u_{j+1})\ge1/2$.
The inequality $k\ge1$ follows from $u\models\neg a$.
At $u_{k-1}$ the formula $\Ne_{>1/2}(\Next\neg a)$ holds.
Its exact next-step condition forces every successor of weight at
least $1/2$, including $u_k$, to satisfy $\neg a$. This is a contradiction.
Thus $\Lambda_\infty$ has no finitely-valued model, and in particular
no finite model.
\end{example}

The main construction below proves that finite satisfiability and
finitely-valued satisfiability coincide. Together with
Example~\ref{ex:no-fmp}, this yields
\[
 \mathrm{SAT}_{\mathrm{finite}}
 =\mathrm{SAT}_{\mathrm{FV}}
 \subsetneq\mathrm{SAT}_{\mathrm{general}}.
\]
The equality is a consequence of the later finite-model theorem, whereas
the strict separation already follows from the example.
All subsequent decision and axiomatization results concern
$\FVPKS$.

\section{Satisfiability over finitely-valued PKSs}\label{sec:sat}
\subsection{Possibilistic Hintikka structures and supported fragments}
\subsubsection{Finite closure and complete types}

In the following, we assume that the candidate formula $\Lambda$ is in positive normal form. Let $V_1(\Lambda)$ denote the set of all real (in general, rational) numbers occurred in the formula $\Lambda$, which is finite.

 Let \begin{equation}
V(\Lambda)=V_1(\Lambda)\cup \{1-r: r\in V_1(\Lambda)\}\cup\{0,1\},
\end{equation}
 which is a finite set closed under maximum, minimum and negation operations.

The {\sl closure} of $\Lambda$, $cl(\Lambda)$, is the least set of subformulae such that

(1) each subformulae of $\Lambda$, including $\Lambda$ itself, is a member of $cl(\Lambda)$;

(2) if $Po_{\sim r}(\lozenge \Phi)$, $Po_{\sim r}(\square \Phi)$, $Po_{\sim r}(\Phi \sqcup \Psi)$ or $Po_{\sim r}(\Phi \Release \Psi)\in cl(\Lambda)$, then, respectively, $Po_{\sim r}(\bigcirc Po_{\sim r}(\lozenge \Phi))$, $Po_{\sim r}(\bigcirc Po_{\sim r}(\square \Phi))$, $Po_{\sim r}(\bigcirc Po_{\sim r}(\Phi \sqcup \Psi))$ or $Po_{\sim r}(\bigcirc$ $Po_{\sim r}(\Phi \Release \Psi))\in cl(\Lambda)$;

(3) if $Ne_{\sim r}(\lozenge \Phi)$, $Ne_{\sim r}(\square \Phi)$, $Ne_{\sim r}(\Phi \sqcup \Psi)$ or $Ne_{\sim r}(\Phi \Release \Psi)\in cl(\Lambda)$, then, respectively, $Ne_{\sim r}(\bigcirc Ne_{\sim r}(\lozenge \Phi))$, $Ne_{\sim r}(\bigcirc Ne_{\sim r}(\square \Phi))$, $Ne_{\sim r}(\bigcirc Ne_{\sim r}(\Phi \sqcup \Psi))$ or $Ne_{\sim r}(\bigcirc$ $Ne_{\sim r}(\Phi \Release \Psi))\in cl(\Lambda)$.

(4) if $Po_{\sim r}(\bigcirc \Phi)$, or $Ne_{\sim r}(\bigcirc \Phi)\in cl(\Lambda)$, then respectively, for any $r'\in V(\Lambda)$ such that $r'\leq r$, $Po_{\sim r'}(\bigcirc \Phi)$, or $Ne_{\sim r'}(\bigcirc \Phi)\in cl(\Lambda)$.

(5) if $Po_{>r}(\bigcirc \Phi)$, or $Ne_{>r}(\bigcirc \Phi)\in cl(\Lambda)$, then respectively, $Po_{\geq r}(\bigcirc \Phi)$, or $Ne_{\geq r}(\bigcirc \Phi)\in cl(\Lambda)$.

(6) if $Po_{\geq r}(\bigcirc \Phi)$, or $Ne_{\geq r}(\bigcirc \Phi)\in cl(\Lambda)$, then respectively, for any $r'\in V(\Lambda)$ such that $r'<r$,$Po_{>r'}(\bigcirc \Phi)$, or $Ne_{>r'}(\bigcirc \Phi)\in cl(\Lambda)$.

%(4) if $Ne_{>r}(\bigcirc \Phi)\in cl(\Lambda)$, then, $Ne_{\geq r}(\bigcirc \Phi)\in cl(\Lambda)$.

The extended closure of $\Lambda$, $ecl(\Lambda)=cl(\Lambda)\cup \{\neg \Phi: \Phi\in cl(\Lambda)\}\cup\{true,false\}\cup\{Po_{\geq 1}(\bigcirc true)\}$. Here, $\neg \Phi$ is assumed in positive normal form.

Since the number of basic subformulas and the number of input thresholds are both at most of the same order as the input length, and the threshold completion forms at most their product, it follows that $|ecl(\Lambda)|=O(|\Lambda|^2)$.

As in CTL, we will distinguish four types of positive normal formulae.

\begin{definition}\label{def:types of PoCTL}

1. {\sl Conjunctive formulae}. Every conjunctive formula, typically denoted by $\alpha$, is associated with several (at most two, in this paper) formulae, called its {\sl conjunctive components}, such that $\alpha$ is equivalent to the conjunction of its conjunctive components.

2.{\sl Disjunctive formulae}. Every disjunctive formula, typically denoted by $\beta$, is associated with several (at most two, in this paper) formulae, called its {\sl disjunctive components}, such that $\beta$ is equivalent to the disjunction of its disjunctive components.

3. {\sl Sucessor formulae}. These are formulae referring to truth in (some, or all) successor states.
There are two types of successor formulae, $Ne_{\sim r}(\bigcirc \Phi)$ and $Po_{\sim r}(\bigcirc \Phi)$. The only component of a successor formulae is its successor component $\Phi$.

4. {\sl Literals:} true, atomic propositions and negations of these. They have no components.

\end{definition}

We summarise the types and components of positive normal formulae in PoCTL in Table \ref{tab:main_table}.

%We say that a formula is {\sl elementary} provided that it is a proposition, the negation of a proposition, or has main connective $Po_{\sim r}(\bigcirc)$ or $Ne_{\sim r}(\bigcirc)$. Any other formula is {\sl non-elementary}. Each non-elementary formula may be viewed as either a conjunctive formula $\alpha=\alpha_1\wedge \alpha_2$ or a disjunctive formula $\beta=\beta_1\vee \beta_2$. A PoCTL formula may be classified as $\alpha$ or $\beta$ formula based on its fixpoint characterization. The following table summarize the classification.

\begin{table}[htbp]
\footnotesize
    \centering
    \caption{Types and components of positive normal formulae in PoCTL}
  \begin{subtable}[t]{\textwidth}
    \centering
  \caption{conjunctive}
  \begin{tabular}{|l|c|}
    \hline
    formula & components  \\
    \hline
    $\Phi\wedge \Psi$ & $\Phi, \Psi$ \\
     $Ne_{\sim r}(\Phi\Release \Psi)$& $\Psi$, $\Phi\vee Ne_{\sim r}(\bigcirc Ne_{\sim r}(\Phi\Release \Psi))$\\
    $Po_{\sim r}(\Phi\Release \Psi)$& $\Psi$, $\Phi\vee Po_{\sim r}(\bigcirc Po_{\sim r}(\Phi\Release \Psi))$\\
    $Ne_{\sim r}(\square \Phi)$& $\Phi$, $Ne_{\sim r}(\bigcirc Ne_{\sim r}(\square \Phi))$\\
    $Po_{\sim r}(\square \Phi)$& $\Phi$, $Po_{\sim r}(\bigcirc Po_{\sim r}(\square \Phi))$\\
    \hline
  \end{tabular}
  \label{subtab:conjunctive}
  \end{subtable}
  \hfill
  \begin{subtable}[t]{\textwidth}
     \centering
  \caption{disjunctive}
  \begin{tabular}{|l|c|}
    \hline
    formula & components  \\
    \hline
    $\Phi\vee \Psi$ & $\Phi, \Psi$ \\
    $Ne_{\sim r}(\Phi \sqcup \Psi)$& $\Psi$, $\Phi\wedge Ne_{\sim r}(\bigcirc (Ne_{\sim r}(\Phi \sqcup \Psi))$\\
    $Po_{\sim r}(\Phi\sqcup \Psi)$& $\Psi$, $\Phi\wedge Po_{\sim r}(\bigcirc (Po_{\sim r}(\Phi \sqcup \Psi))$\\
    $Ne_{\sim r}(\lozenge \Phi)$& $\Phi$, $Ne_{\sim r}(\bigcirc Ne_{\sim r}(\lozenge \Phi))$\\
    $Po_{\sim r}(\lozenge \Phi)$& $\Phi$, $Po_{\sim r}(\bigcirc (Po_{\sim r}(\lozenge \Phi))$\\
    \hline
  \end{tabular}
  \label{subtab:disjunctive}
  \end{subtable}
  \hfill
  \begin{subtable}[t]{\textwidth}
     \centering
  \caption{successor}
  \begin{tabular}{|l|c|}
    \hline
    formula & components  \\
    \hline
    $Ne_{\sim r}(\bigcirc \Phi)$& $\Phi$\\
    $Po_{\sim r}(\bigcirc \Phi)$& $\Phi$\\
    \hline
  \end{tabular}
  \label{subtab:successor}
  \end{subtable}

  \label{tab:main_table}
\end{table}

A formula of the form $Ne_{\sim r}(\Phi \sqcup \Psi)$ or $Po_{\sim r}(\Phi \sqcup \Psi)$, and its special case $Ne_{\sim r}(\lozenge \Phi)$ or $Po_{\sim r}(\lozenge \Phi)$, is an {\sl eventuality} formula. An eventuality formula makes a promise that something will happen. This promise must be fulfilled (see Eq.(\ref{eq:cut-po-ge}-\ref{eq:cut-ne-ge})). In contrast, $Ne_{\sim r}(\Phi \Release \Psi)$ or $Po_{\sim r}(\Phi \Release \Psi)$, and their special case $Ne_{\sim r}(\square \Phi)$ and $Po_{\sim r}(\square \Phi)$, are {\sl invariant} formula. An invariant property asserts that whatever happens to occur (if anything) will meet certain conditions.

We first state the \emph{propositional consistency rules (PCR)} used
in the original tableau construction. A subset $s\subseteq\ecl(\Lambda)$
is \emph{propositionally consistent} if it satisfies
\begin{align*}
\text{(PC0)}\quad&\neg\Phi\in s\ \Longrightarrow\ \Phi\notin s,\\
\text{(PC1)}\quad&\alpha\in s\ \Longrightarrow\
                  \alpha_1\in s\ \text{and}\ \alpha_2\in s,\\
\text{(PC2)}\quad&\beta\in s\ \Longrightarrow\
                  \beta_1\in s\ \text{or}\ \beta_2\in s.
\end{align*}
In (PC1), $\alpha$ has the conjunctive form $\alpha_1\wedge\alpha_2$,
either syntactically, through its unfolding in Table~\ref{tab:main_table},
or through one of the threshold decompositions below.
Likewise, (PC2) applies to the disjunctive form $\beta_1\vee\beta_2$.
Thus the rules apply both to Boolean formulas and to temporal formulas
through their specified decompositions. Complements are understood in
positive normal form. A formula may have several specified conjunctive
decompositions; (PC1) applies to each of them.

For next formulas, the following threshold absorption identities specify
conjunctive decompositions for (PC1):
\begin{align*}
 \Po_{>r}(\Next\Phi)
 &\equiv \Po_{>r}(\Next\Phi)\wedge\Po_{\ge r}(\Next\Phi),\\
 \Po_{\sim r}(\Next\Phi)
 &\equiv \Po_{\sim r}(\Next\Phi)\wedge\Po_{\sim r'}(\Next\Phi)
       &&(r'<r),\\
 \Po_{\ge r}(\Next\Phi)
 &\equiv \Po_{\ge r}(\Next\Phi)\wedge\Po_{>r'}(\Next\Phi)
       &&(r'<r).
\end{align*}
The same decompositions are specified with $\Ne$ in place of $\Po$.
Here $r,r'\in V(\Lambda)$, and only nondegenerate bounds are used.
Threshold completion ensures that every component on the right belongs
to $\ecl(\Lambda)$ whenever the formula on the left does.
These are component prescriptions, so the displayed conjunctions do
not have to be added as new elements of the closure.
Only the explicitly specified decompositions are used; no closure under
arbitrary semantic equivalence is assumed.

Following the original tableau terminology, such a subset is called
\emph{maximal} when, for every $\Phi\in\ecl(\Lambda)$, either
$\Phi\in s$ or $\neg\Phi\in s$; (PC0) makes these alternatives exclusive.

\begin{definition}[Complete tableau type]\label{def:type}
A complete tableau type is a maximal propositionally consistent subset
$s$ of $\ecl(\Lambda)$ that contains $\true$ and
$\Po_{\ge1}(\Next\true)$ and excludes $\false$.
\end{definition}

Maximality, together with (PC0--2) applied also to complements, enforces
the Boolean equivalences represented by parts (a)--(b) of
Table~\ref{tab:main_table}.
Threshold consistency follows from (PC1) using the absorption
decompositions. For example, if $\Ne_{>r}(\Next\Phi)\in s$ and
$r'<r$, its specified conjunctive components include
$\Ne_{>r'}(\Next\Phi)$, which therefore belongs to $s$.
The other decompositions yield strict-to-nonstrict weakening and
passage from a non-strict lower bound to a strict lower bound at a
strictly smaller threshold. Thus no separate threshold-consistency
condition is needed in Definition~\ref{def:type}.
These are finite local checks, rather than consistency with respect
to all theorems of the axiomatic system. The prescribed decompositions
are provable from the axioms: the threshold identities use AP1--AP3,
duality, and propositional logic, while the temporal unfoldings are
derived independently in Lemma~\ref{lem:replacement}.

A {\sl prestructure} $M$ for a formula $\Lambda$ is a triple $(S, P, \mathscr{L})$ just like a PKS, { except that no normal
condition is imposed on the possibilistic transition relation $P:S\times S\to[0,1]$, and the label function $\mathscr{L}:S\to 2^{ecl(\Lambda)}$
is a formula-labeling function.} Furthermore, if the normal condition holds on possibilistic transition relation $P$, then $M$ is called a {\sl structure} for $\Lambda$.  An {\sl interior} node of a prestructure is one with at least one successor, i.e., those $s$ such that $P(s,t)>0$ or { equivalently} $P(s,t)\geq r$ for some state $t$ and $r\in (0,1]$. A {\sl frontier} node is one with no successors, i.e., those $s$ such that ${\rm sup}\{P(s,t): t\in S\}=0$.

{ To avoid confusion, throughout this section we use
$L:S\to 2^{AP}$ for the atomic-proposition labeling
function of a possibilistic Kripke structure, and
$\mathscr{L}:S\to 2^{ecl(\Lambda)}$
for the formula-labeling function of a prestructure.}

\subsubsection{Local conditions and fulfillment}

For nondegenerate thresholds, the local consistency requirements are:
\begin{align*}
\text{(LC0-1)}\quad&
 \Ne_{>r}(\Next\Phi)\in\mathcal L(s),\
 P(s,t)\ge1-r\ \Longrightarrow\ \Phi\in\mathcal L(t),\\
\text{(LC0-2)}\quad&
 \Ne_{\ge r}(\Next\Phi)\in\mathcal L(s),\
 P(s,t)>1-r\ \Longrightarrow\ \Phi\in\mathcal L(t),\\
\text{(LC1-1)}\quad&
 \Po_{>r}(\Next\Phi)\in\mathcal L(s)\
 \Longrightarrow\
 \exists t\ [P(s,t)>r,\ \Phi\in\mathcal L(t)],\\
\text{(LC1-2)}\quad&
 \Po_{\ge r}(\Next\Phi)\in\mathcal L(s)\
 \Longrightarrow\
 \exists t\ [P(s,t)\ge r,\ \Phi\in\mathcal L(t)].
\end{align*}
In particular, the normalization formula requires a weight-$1$ successor.
Merely having a positive successor is insufficient.

An eventuality is an until formula
$\Po_{\sim r}(\Phi\Until\Psi)$ or
$\Ne_{\sim r}(\Phi\Until\Psi)$, including eventually as a special case.
It is fulfilled at a state when the corresponding condition
\eqref{eq:cut-po-ge}--\eqref{eq:cut-ne-ge} holds with formula labels in
place of semantic truth: the target label $\Psi$ is reached after finitely
many steps, with $\Phi$ at all preceding states.
A finitely-valued structure meeting the complete-type conditions,
(LC0--1), and these eventuality requirements is a
\emph{possibilistic Hintikka structure}.
The local release unfoldings already express the relevant invariant
requirements.

\begin{definition}[Supported fragment]\label{def:fragment}
A fragment is a finite directed acyclic graph with a root, whose nodes
carry labels from a given prestructure. Its frontier consists of its
nodes with no outgoing edges. Every interior node satisfies all its
local possibility-next requirements, including the weight-$1$
normalization requirement. Every edge satisfies (LC0).

If an edge joins a copy of $s$ to a copy of $t$, its weight may be any
positive value not exceeding $P(s,t)$ that satisfies these conditions.
Thus a fragment selects admissible edges and may lower their weights.
Different copies of the same label are allowed.

A fragment $\DAG[s,E]$ certifies an eventuality
$E\in\mathcal L(s)$ as follows.
\begin{enumerate}[label=(\roman*),leftmargin=2.3em]
\item If $E=\Po_{\sim r}(\Phi\Until\Psi)$, it contains a designated
path from the root to a $\Psi$-node, with $\Phi$ at all preceding
nodes and with all prefix edges relevant to $E$.
Other local support edges may end at frontier nodes.
\item If $E=\Ne_{\sim r}(\Phi\Until\Psi)$, every root-to-frontier
path using only relevant edges reaches a $\Psi$-node, and all nodes
preceding that target carry $\Phi$.
Edges outside the relevant threshold may end at frontier nodes
without completing $E$.
\end{enumerate}
A root already containing $\Psi$ has a trivial certifying fragment.
\end{definition}

\begin{remark}\label{rem:low-exit}
The frontier condition is relative to the current threshold.
For example, take $E=\Ne_{>0}(\Diamond a)$ and the finite PKS
\[
 P(u,v)=1,\quad P(u,w)=\frac12,\quad P(v,v)=P(w,w)=1,
\]
with $a$ true only at $v$. Then
$u\models E\wedge\Po_{>0}(\Next\neg E)$.
A supported fragment at $u$ must retain a low-weight witness for
$\Next\neg E$. That edge is irrelevant to the weight-$1$ paths
constrained by $E$ and can leave the fragment immediately.
Requiring all frontier nodes, including such low-weight exits,
to complete $E$ would impose an unwarranted condition.
\end{remark}

\begin{definition}[Finite pseudo-possibilistic Hintikka structure]\label{def:pseudo}
A finite prestructure for $\Lambda$ is a pseudo-possibilistic Hintikka
structure if some node contains $\Lambda$, all nodes satisfy the
complete-type and local conditions above, and every eventuality in
every node label has a finite certifying supported fragment.
\end{definition}

\subsubsection{Cyclic splicing of finite fragments}

\begin{theorem}[Finite fragment realization]\label{thm:splice}
A formula has a finite pseudo-possibilistic Hintikka structure if and
only if it has a finite PKS model.
\end{theorem}
\begin{proof}
First assume that $M=(S,P,\mathcal L)$ is a finite pseudo-possibilistic
Hintikka structure. Enumerate its eventualities as $E_1,\ldots,E_m$.
For every node $s$ and eventuality $E$, form a fragment
$\DAGG[s,E]$ rooted at a copy of $s$.
If $E\in\mathcal L(s)$ and its target has not yet been reached,
use $\DAG[s,E]$.
If $E$ is absent or its target already holds, use a one-step fragment
containing a sufficient set of successors for all local requirements.
Thus every $\DAGG$ contains at least one edge.

Arrange these fragments in a matrix whose rows range over
$E_1,\ldots,E_m$ and whose columns range over the nodes of $S$.
Every frontier node with label $\mathcal L(t)$ in row $i$
is attached to the root of the corresponding
$\DAGG[t,E_{i+1}]$ in the next row.
The last row is attached to the first.
All selected edge weights are preserved.
If there are no eventualities, the one-step fragments alone are
attached at matching roots.

The resulting finite graph may temporarily have two different
weights on the same ordered pair of vertices. To obtain an ordinary
transition function, retain a separate target copy for each possible
incoming weight. An edge of weight $q$ is directed to the corresponding
copy of its target, and every copy duplicates the same outgoing edges.
Only finitely many weights occur, so this adds finitely many nodes.
It preserves exactly the possible sequences of labels and edge weights.

Every resulting node has a weight-$1$ successor. Hence the graph,
with atomic labels read from $\mathcal L$, defines a normalized finite
PKS $N=(W,Q,\ell)$.
Each infinite path passes through infinitely many fragments: a path
cannot remain forever within a finite DAG. It consequently enters
every row infinitely often. This remains true for a single row,
where every frontier attachment begins another finite fragment.
When $m>0$, an exit from row $i$ is an entry at a root in row
$i+1$, with indices read cyclically. Each complete passage through a
fragment takes a positive finite number of edges. Thus, from any node,
the next entry to any prescribed row occurs after finitely many steps.
The incoming-weight copies preserve this statement under projection
to the spliced graph. The one-step fragments used for absent or
already fulfilled eventualities exclude zero-length cycles of attachments.

We prove that every formula in the label of a node is true there.
The argument is by induction on formulas. Atomic and Boolean cases
follow from the complete-type conditions.
For a possibility-next formula, a retained local support edge
provides the required threshold and target label.
For a necessity-next formula, (LC0) applies to every retained relevant
edge. The induction hypothesis at the successor and
Lemma~\ref{lem:fv} give the corresponding semantic claims.

Suppose $E=\Po_{\sim r}(\Phi\Until\Psi)$ belongs to the label of the
current node. Until $\Psi$ occurs, its unfolding supplies a support
edge for $\Po_{\sim r}(\Next E)$.
Repeatedly choose such an edge. The target continues to contain $E$,
the preceding nodes contain $\Phi$, and every chosen edge is relevant.
If this process did not reach $\Psi$, it would generate an infinite
path. At a later entry to the row for $E$, the root would still
contain $E$. Follow the designated witnessing prefix of that
fragment instead. This reaches $\Psi$ after finitely many relevant
edges, with $\Phi$ before it. Extend from the target along
weight-$1$ edges. By Lemma~\ref{lem:fv}, this is an until witness
at the required possibility threshold.

Now suppose $E=\Ne_{\sim r}(\Phi\Until\Psi)$ belongs to the current
label. Consider an arbitrary relevant infinite path.
As long as $\Psi$ has not occurred, the unfolding yields $\Phi$
and the necessity-next formula for $E$.
Condition (LC0) therefore propagates $E$ to the next label.
If the path avoided $\Psi$ forever, it would eventually enter the
row for $E$ at a root still containing $E$.
Its passage through that fragment uses only relevant edges,
so the fragment condition requires a $\Psi$-node before or at its
frontier. This is a contradiction.
Before the first such target, all labels contain $\Phi$.
The path was arbitrary, and Lemma~\ref{lem:fv} proves the necessity
until formula.

For possibility release, the unfolding guarantees $\Psi$.
If $\Phi$ is not yet present, it supplies a relevant next support
for the same release formula. Repeatedly choose this support.
Either $\Phi$ is reached while $\Psi$ has held throughout, or $\Psi$
holds forever. The first case can be extended along weight-$1$
edges. Both cases satisfy release at the required threshold.
Necessity release is proved by the corresponding propagation
along every relevant path.
Eventually and always are the stated special cases.

Only proper subformulas $\Phi,\Psi$ are handled by the semantic
induction hypothesis in the temporal cases. Repeated occurrences
of the same eventuality are propagated by local label rules,
so the argument is not circular.
If a formula is absent from a complete label, its complement is
present and hence true. This proves the full truth equivalence
and, in particular, realizes $\Lambda$.

Conversely, let a finite PKS satisfy $\Lambda$.
Label each state by the formulas in $\ecl(\Lambda)$ true there.
Lemma~\ref{lem:fv} and the unfoldings give all local conditions.
For a possibility until formula, choose a shortest finite witnessing
prefix and add local supports at its interior nodes; all additional
branches can end at frontier copies.
For a necessity until formula, consider the relevant states reachable
before the target. They cannot contain a directed cycle avoiding
$\Psi$, since repeating that cycle would yield an infinite relevant
counterexample path. Consequently the relevant part is acyclic.
Include it up to the target, and terminate low-weight local support
edges at frontier copies. These are finite supported fragments.
Thus the labeled structure is a finite pseudo-possibilistic Hintikka
structure.
\end{proof}

\subsection{The finite threshold scale and initial tableau}\label{sec:D}
Write
\[
 V(\Lambda)=\{0=t_1<t_2<\cdots<t_k=1\}.
\]
Insert one midpoint between each pair of consecutive values and put
\begin{equation}\label{eq:EV}
 EV(\Lambda)=V(\Lambda)\cup
     \left\{\frac{t_i+t_{i+1}}2:1\le i<k\right\}.
\end{equation}
Both sets are invariant under $q\mapsto1-q$.
Indeed, if $t_i<t_{i+1}$ are consecutive in $V(\Lambda)$,
then $1-t_{i+1}<1-t_i$ are consecutive as well, and the reflection
of their midpoint is the midpoint of this reflected pair.
For a value of $EV(\Lambda)$, write $r^-$ and $r^+$ for its immediate
predecessor and successor, when they exist.

Let $S_0$ be the set of all complete types over $\ecl(\Lambda)$.
For $s,t\in S_0$, define
\begin{align}
 D(s,t)={}&
 \{1-r:r\in V(\Lambda),\
 \forall\Ne_{>r}(\Next\Phi)\in\ecl(\Lambda),\notag\\[-2pt]
 &\hspace{32mm}
 [\Ne_{>r}(\Next\Phi)\in s\Rightarrow\Phi\in t]\}\notag\\
 &\cup\{1-r^-:r\in V(\Lambda)\setminus\{0\},\
 \forall\Ne_{\ge r}(\Next\Phi)\in\ecl(\Lambda),\notag\\[-2pt]
 &\hspace{32mm}
 [\Ne_{\ge r}(\Next\Phi)\in s\Rightarrow\Phi\in t]\}.
 \label{eq:D}\\
 P(s,t)={}&\max D(s,t),\qquad \mathcal L(s)=s.
 \label{eq:P}
\end{align}
The original two-part construction is retained.
The first part represents a non-strict transition threshold;
the second part represents the next available value above a strict
threshold. In particular,
\[
 r^+=1-(1-r)^-\qquad(r\in V(\Lambda),\ r<1).
\]

\begin{proposition}[Admissible weights]\label{prop:D}
The set $D(s,t)$ contains $0$ and is a downward-closed subset of
$EV(\Lambda)$. In particular, $P(s,t)=\max D(s,t)$ is well defined.
Moreover, for $q\in EV(\Lambda)$, membership $q\in D(s,t)$ is equivalent
to the following two requirements:
\begin{enumerate}[label=(\roman*),leftmargin=2.3em]
\item if $\Ne_{>r}(\Next\Phi)\in s$ and $q\ge1-r$, then $\Phi\in t$;
\item if $\Ne_{\ge r}(\Next\Phi)\in s$ and $q>1-r$, then $\Phi\in t$.
\end{enumerate}
Consequently, the initial tableau satisfies (LC0).
\end{proposition}
\begin{proof}
At $r=1$, the first-part test in \eqref{eq:D} uses $\Ne_{>1}$,
whose instances have been simplified to $\false$ and cannot belong
to $s$. The test is vacuous, so $0\in D(s,t)$.
Reflection invariance of the scale gives $D(s,t)\subseteq EV(\Lambda)$.

We first prove downward closure of the set $D(s,t)$. Let
$q\in D(s,t)$ and $q'\in EV(\Lambda)$ with $q'\le q$.
The case $q'=0$ is already settled, so assume $0<q'\le q$.
Every value of $EV(\Lambda)$ is of the form $1-r$ or $1-r^-$,
where $r\in V(\Lambda)$ and $r>0$ whenever $r^-$ occurs.
The threshold decompositions for (PC1) are available in every case
by the definition of $\ecl(\Lambda)$.

\smallskip\noindent
\emph{Case 1: $1-r\in D(s,t)$ and $1-r'\le1-r$.}
Then $r'\ge r$. Take any $\Ne_{>r'}(\Next\Phi)\in s$.
By (PC1) and its threshold decomposition,
$\Ne_{>r}(\Next\Phi)\in s$; when $r'=r$ this is the same formula.
Since $1-r\in D(s,t)$, the first part of \eqref{eq:D} gives
$\Phi\in t$. Thus every $\Ne_{>r'}(\Next\Phi)\in s$ has
$\Phi\in t$, and $1-r'\in D(s,t)$.

\smallskip\noindent
\emph{Case 2: $1-r\in D(s,t)$ and $1-(r')^-\le1-r$.}
Here $(r')^-\ge r$, and hence $r'>r$.
Take any $\Ne_{\ge r'}(\Next\Phi)\in s$.
The mixed-threshold decomposition and (PC1) give
$\Ne_{>r}(\Next\Phi)\in s$.
The first part of \eqref{eq:D} then gives $\Phi\in t$.
This is exactly the second-part test for the value $1-(r')^-$,
so $1-(r')^-\in D(s,t)$.

\smallskip\noindent
\emph{Case 3: $1-r^-\in D(s,t)$ and $1-r'\le1-r^-$.}
The inequality gives $r'\ge r^-$. Since $r'\in V(\Lambda)$
and $r^-$ is the midpoint immediately preceding $r$, there is no
value of $V(\Lambda)$ in $[r^-,r)$. Thus $r'\ge r$.
Take any $\Ne_{>r'}(\Next\Phi)\in s$.
Strict-to-nonstrict weakening and, if necessary, a decrease of the
lower bound, both through (PC1), give $\Ne_{\ge r}(\Next\Phi)\in s$.
Since $1-r^-\in D(s,t)$, the second part of \eqref{eq:D} gives
$\Phi\in t$. The first-part test therefore holds at $1-r'$,
so $1-r'\in D(s,t)$.

\smallskip\noindent
\emph{Case 4: $1-r^-\in D(s,t)$ and $1-(r')^-\le1-r^-$.}
We have $(r')^-\ge r^-$, which implies $r'\ge r$:
if $r'<r$, their immediate predecessors in the ordered scale would
satisfy $(r')^-<r^-$.
For any $\Ne_{\ge r'}(\Next\Phi)\in s$, (PC1) gives
$\Ne_{\ge r}(\Next\Phi)\in s$.
The second-part test at $1-r^-$ yields $\Phi\in t$,
so that test also holds at $1-(r')^-$.
Hence $1-(r')^-\in D(s,t)$.

These cases prove downward closure. Since $EV(\Lambda)$ is finite
and $0\in D(s,t)$, the maximum of $D(s,t)$ exists.

We now prove the stated characterization. Suppose $q\in D(s,t)$.
If $\Ne_{>r}(\Next\Phi)\in s$ and $q\ge1-r$, downward closure gives
$1-r\in D(s,t)$, whose first-part test yields $\Phi\in t$.
If $\Ne_{\ge r}(\Next\Phi)\in s$ and $q>1-r$, discreteness and
reflection of the scale give
\[
 q\ge(1-r)^+=1-r^-.
\]
Downward closure now gives $1-r^-\in D(s,t)$, whose second-part
test yields $\Phi\in t$. Thus (i) and (ii) hold.

Conversely, assume (i) and (ii). If $q=0$, membership is already known.
If $q=1-r>0$ with $r\in V(\Lambda)$, condition (i) is precisely
the first-part test at $q$.
If $q=1-r^-$ with $r\in V(\Lambda)\setminus\{0\}$,
then $q>1-r$, so condition (ii) supplies the second-part test.
In both cases $q\in D(s,t)$.
Finally, applying this characterization to $P(s,t)=\max D(s,t)$
proves (LC0).
\end{proof}

The threshold implications used above follow from (PC1) through the
absorption decompositions specified before Definition~\ref{def:type}.
Threshold completion ensures that all their components are available
in $\ecl(\Lambda)$.

\begin{proposition}[Real transitions provide admissible type edges]\label{prop:actual}
Let $N=(W,Q,\ell)\in\FVPKS$, and let
$s=\tp_\Lambda(u)$ and $t=\tp_\Lambda(v)$ be its true types.
For $r\in V(\Lambda)$ at a nondegenerate bound,
\[
 Q(u,v)\ge r\ \Longrightarrow\ P(s,t)\ge r,
 \qquad
 Q(u,v)>r\ \Longrightarrow\ P(s,t)\ge r^+>r.
\]
\end{proposition}
\begin{proof}
If $Q(u,v)\ge r$, every $\Ne_{>1-r}(\Next\Phi)$ in $s$
forces $\Phi$ at $v$, by Lemma~\ref{lem:fv}.
The first part of $D$ therefore contains $r$.
If $Q(u,v)>r$, the same argument uses
$\Ne_{\ge1-r}(\Next\Phi)$ and the second part of $D$,
which contains $r^+=1-(1-r)^-$.
\end{proof}

The second implication concerns the abstract weight $P(s,t)$.
It does not assert that the original weight $Q(u,v)$ is at least $r^+$.
Nor does the tableau require every edge of maximum admissible weight
to be retained in the final model.

\subsection{Independent eventuality ranks and deletion}\label{sec:ranks}

Let $S_i$ be the current set of tableau nodes.
For a requirement $\Po_{\ge r'}(\Next\Xi)\in s$ its allowed successors are
\[
 \{t\in S_i:\Xi\in t,\ P(s,t)\ge r'\},
\]
and for $\Po_{>r'}(\Next\Xi)\in s$ they are
\[
 \{t\in S_i:\Xi\in t,\ P(s,t)>r'\}.
\]
When a witness is selected for the first requirement, give its edge
weight $r'$; for the second, give it weight $r'^+$.
Both weights are admissible by Proposition~\ref{prop:D}.
The normalization requirement has selected weight $1$.

\subsubsection{Possibility eventualities}

For $E=\Po_{\ge r}(\Phi\Until\Psi)$, initialize
\[
 \rank_E(s)=
 \begin{cases}
 1,&\Psi\in s,\\
 \infty,&\Psi\notin s,
 \end{cases}
 \qquad
 \operatorname{Succ}_E(s)=\{t\in S_i:P(s,t)\ge r\}.
\]
At a node of infinite rank with $\Phi\in s$, assign a finite value
whenever the following right-hand side is finite:
\begin{equation}\label{eq:rank-po}
 \rank_E(s)=1+
     \min\{\rank_E(t):t\in\operatorname{Succ}_E(s)\}.
\end{equation}
Use simultaneous rounds, reading only the preceding round's values.
A finite rank, once assigned, is unchanged. Iterate until stabilization
and put $\min\varnothing=\infty$.
For $E=\Po_{>r}(\Phi\Until\Psi)$, replace the successor condition
by $P(s,t)>r$.

\begin{proposition}[The original min rank already preserves the eventuality]\label{prop:po-label}
For either possibility eventuality,
$\rank_E(s)<\infty$ implies $E\in s$.
At a non-target node of finite rank, a lower-rank relevant successor
therefore also contains $E$.
\end{proposition}
\begin{proof}
Induct on the round in which a finite rank is assigned.
At rank $1$, $\Psi\in s$, and the unfolding implies $E\in s$.

Suppose $\Psi\notin s$ and $s$ receives a finite rank.
Then $\Phi\in s$ and a relevant successor $t$ has already received
a finite rank. By induction, $E\in t$.
If $E\notin s$, completeness gives $\neg E\in s$.
For $E=\Po_{\ge r}(\Phi\Until\Psi)$, its complementary unfolding yields
\[
 \neg E\ \to\
 [\neg\Psi\wedge(\Phi\to\Ne_{>1-r}(\Next\neg E))].
\]
Thus $\Ne_{>1-r}(\Next\neg E)\in s$.
Since $P(s,t)\ge r$, (LC0) forces $\neg E\in t$, a contradiction.
The strict case uses $\Ne_{\ge1-r}(\Next\neg E)$ and $P(s,t)>r$.
Hence $E\in s$ in both cases.
The lower-rank successor is supplied by \eqref{eq:rank-po}.
\end{proof}

No extra membership filter is therefore needed in
$\operatorname{Succ}_E(s)$. In particular, the infinite-rank set
used in the axiomatic refutation of a possibility eventuality can
include \emph{all} infinite-rank nodes.

\subsubsection{Necessity eventualities}

Initialize the ranks for
$E=\Ne_{>r}(\Phi\Until\Psi)$ and
$E=\Ne_{\ge r}(\Phi\Until\Psi)$ in the same way.
The maximum in the rank update must range over all local
possibility-next requirements whose selected weights are relevant to $E$.
These are exactly the requirements in Table~\ref{tab:relevant}.

\begin{table}[htbp]
\centering
\caption{Requirements included in the necessity max--min rank.}\label{tab:relevant}
\small
\begin{tabular}{lll}
\toprule
Eventuality & $\Po_{\ge r'}(\Next\Xi)\in s$ & $\Po_{>r'}(\Next\Xi)\in s$\\
\midrule
$\Ne_{>r}(\Phi\Until\Psi)$ & $r'\ge1-r$ & $r'\ge1-r$\\
$\Ne_{\ge r}(\Phi\Until\Psi)$ & $r'>1-r$ & $r'\ge1-r$\\
\bottomrule
\end{tabular}
\end{table}

To see the strict cases, recall that $1-r\in V(\Lambda)$.
If $r'<1-r$, then $r'^+<1-r$.
If $r'\ge1-r$, then $r'^+>1-r$.
Together with the selected weight $r'$ for a non-strict requirement,
these observations give precisely the table.
The normalization requirement is always included.

At a node of infinite rank with $\Phi\in s$, the update is
\begin{equation}\label{eq:rank-ne}
 \rank_E(s)=1+
 \max_{\substack{\Po_{\sim r'}(\Next\Xi)\in s\\
                 \text{included in Table~\ref{tab:relevant}}}}
 \ \min_{\substack{t\in S_i,\ \Xi\in t\\P(s,t)\sim r'}}
       \rank_E(t).
\end{equation}
The maximum is over a nonempty finite set because of normalization.
This is the original max--min construction with each
$N_\Xi^E(s)$ written out explicitly.
The successor set is evaluated at the requirement's own bound $\sim r'$,
not merely at the threshold of $E$.
As before, use simultaneous rounds until stabilization.

\begin{proposition}[What the max--min rank guarantees]\label{prop:ne-rank}
If $E\in s$, $\Psi\notin s$, and $\rank_E(s)<\infty$,
then each requirement in Table~\ref{tab:relevant} has an allowed
successor of smaller $E$-rank.
All these successors contain $E$.
\end{proposition}
\begin{proof}
Each minimum in \eqref{eq:rank-ne} is finite and strictly below
$\rank_E(s)$. Choose a successor attaining it.
The unfolding of $E$ at $s$ gives $\Phi$ and the corresponding
necessity-next formula for $E$.
The selected successor is relevant, so (LC0) transmits $E$ to it.
\end{proof}

This proposition is a statement about choosing one suitable successor
for each requirement. It does not imply that every admissible successor
has smaller rank, or that the graph containing all maximum-weight
tableau edges already fulfills $E$.

\subsubsection{Termination and deletion}

In a current tableau with $N$ nodes, every finite rank is at most $N$:
each nonstationary assignment round adds a finite-rank node, and
after a stationary round no further assignments can occur.
Starting with target rank $1$, induction shows that a rank first
assigned in round $j$ is at most $j+1$; a finite rank requires at
least one target. These observations give the stated bound.
Also, enlarging the current node set only enlarges allowed successor sets.
Induction on assignment rounds shows that nodes already assigned a
finite rank remain finitely ranked in a larger tableau.

Starting at $S_0$, repeatedly recompute the ranks and apply:
\begin{enumerate}[label=(2-\arabic*),leftmargin=3em]
\item Delete any node without a weight-$1$ admissible successor.
\item Delete any node whose local possibility-next requirement
      has no allowed successor at its own bound.
\item Delete any node containing an eventuality $E$ whose rank is infinite.
\end{enumerate}
Rule (2-1) can equivalently be absorbed into (2-2) using
$\Po_{\ge1}(\Next\true)$.
Continue until no more nodes are deleted.
Accept $\Lambda$ exactly when a remaining node contains it.
Each rank is computed independently; there is no deletion rule requiring
the ranks of different eventualities to decrease together.

\subsection{Correctness of the tableau procedure}\label{sec:correctness}
\subsubsection{From finite ranks to supported fragments}

\begin{proposition}[Fragments selected by independent ranks]\label{prop:rank-fragments}
If no deletion rule applies to a finite tableau, every eventuality
in every remaining node has a finite certifying supported fragment.
\end{proposition}
\begin{proof}
All local allowed successor sets are nonempty.
Fix $s$ and an eventuality $E\in s$.
If its target $\Psi$ is already present, the trivial fragment suffices.

Let $E=\Po_{\sim r}(\Phi\Until\Psi)$ with $\Psi\notin s$.
At each interior node, select for the requirement
$\Po_{\sim r}(\Next E)$ a relevant successor of strictly smaller
$E$-rank. Proposition~\ref{prop:po-label} shows that this successor
contains $E$, so the selected edge really supports that requirement.
Give the edge weight $r$ or $r^+$, according to the comparison.
Continue at its target unless $\Psi$ is present.
For every other local possibility-next requirement, select an allowed
successor at its own bound and terminate that branch at a frontier copy.
This includes normalization. The designated path strictly decreases
the positive integer rank and therefore ends at a rank-$1$ target.
Its preceding nodes contain $\Phi$.

For $E=\Ne_{\sim r}(\Phi\Until\Psi)$, proceed differently at each
non-target interior node. For every requirement listed in
Table~\ref{tab:relevant}, choose a lower-$E$-rank allowed successor,
using Proposition~\ref{prop:ne-rank}.
Use its prescribed weight $r'$ or $r'^+$.
If the successor contains $\Psi$, terminate that branch;
otherwise continue it in the fragment.
Every requirement not listed in the table is also supported,
but its selected edge terminates immediately at a frontier copy.
By the table, these additional edges are irrelevant to $E$.
Thus every relevant edge continuing inside the fragment decreases rank
and propagates $E$, while $\Phi$ holds before the target.
A root-to-frontier path of relevant edges must reach $\Psi$.

In both cases, all interior nodes support all local requirements.
Every selected weight belongs to $D$, so necessity-next conditions
remain true by Proposition~\ref{prop:D}.
Finite ranks bound the depth, and the number of requirements bounds
the branching, hence the construction is finite.
For a smaller representation, fix the choices once for each label
within the fragment and share equal-labeled interior nodes.
Each interior edge still decreases the same rank, so sharing introduces
no directed cycle. Frontier copies are kept separate.
\end{proof}

The proposition explains why independent ranks are sufficient.
The fragment for $E$ makes progress on $E$ and preserves local
supports for all other formulas. The cyclic construction of
Theorem~\ref{thm:splice} then completes those other eventualities
in their own rows.

\subsubsection{From arbitrary finitely-valued models to surviving types}

\begin{proposition}[Preservation of true types]\label{prop:preservation}
Let $N=(W,Q,\ell)\in\FVPKS$, and put
\[
 \tp_\Lambda(u)=\{\Phi\in\ecl(\Lambda):N,u\models\Phi\},
 \qquad
 T_N=\{\tp_\Lambda(u):u\in W\}.
\]
No node in $T_N$ is deleted by the tableau procedure.
\end{proposition}
\begin{proof}
Every true type belongs to $S_0$, by the semantic unfoldings, threshold
monotonicity, and normalization.
The set $T_N$ is finite because the closure is finite; the original
state set $W$ need not be finite.
First run the local and rank checks on $T_N$ itself.

For a local requirement $\Po_{\sim r'}(\Next\Xi)$ at
$s=\tp_\Lambda(u)$, Lemma~\ref{lem:fv} supplies an actual successor
$v$ satisfying $\Xi$ with $Q(u,v)\sim r'$.
Proposition~\ref{prop:actual} makes its type an allowed successor
in $T_N$. A weight-$1$ actual successor witnesses normalization.
Thus the first two deletion rules cannot remove a true type.

If a possibility eventuality
$E=\Po_{\sim r}(\Phi\Until\Psi)$ is true at $u$,
Lemma~\ref{lem:fv} supplies a finite relevant prefix
$u=u_0,u_1,\ldots,u_k$ satisfying the until requirement.
Proposition~\ref{prop:actual} transfers its edges to the appropriate
admissible type edges.
The target type has rank $1$; backwards application of the min rule
gives
$\rank_E(\tp_\Lambda(u))\le k+1<\infty$.

Consider $E=\Ne_{>r}(\Phi\Until\Psi)$.
Suppose its ranks have stabilized in $T_N$ and the set
\[
 T_E=\{s\in T_N:E\in s,\ \rank_E(s)=\infty\}
\]
is nonempty.
Every $s\in T_E$ contains $\Phi$ and
$\Ne_{>r}(\Next E)$, and does not contain $\Psi$.

For each such $s$, some relevant local requirement has only
infinite-$E$-rank allowed successors.
Indeed, if every relevant requirement had a finite-rank successor,
the finite maximum of their finite minimum ranks would give
$s$ a finite rank in \eqref{eq:rank-ne}.
Every allowed successor of the selected requirement also contains $E$:
its edge is relevant and (LC0) applies to
$\Ne_{>r}(\Next E)\in s$.
All these allowed successors therefore lie in $T_E$.

Choose $u_0$ with $\tp_\Lambda(u_0)\in T_E$.
Having chosen $u_j$, use the actual witness for the requirement
selected at its type to choose $u_{j+1}$.
Its type remains in $T_E$, and the requirement's own threshold implies
\[
 Q(u_j,u_{j+1})\ge1-r.
\]
This produces an infinite actual relevant path whose every node
avoids $\Psi$, contradicting $E$ at $u_0$.

For $E=\Ne_{\ge r}(\Phi\Until\Psi)$, use the second row of
Table~\ref{tab:relevant}.
The selected non-strict requirements have $r'>1-r$, and the
selected strict ones have $r'\ge1-r$.
Their actual witnesses therefore satisfy
$Q(u_j,u_{j+1})>1-r$.
Again the target types remain in $T_E$.
Since the actual transition values form a finite set, the resulting
infinite path has possibility strictly greater than $1-r$.
It is a counterexample to $E$, a contradiction.
This step uses the inequalities on the actual requirements, not an
assertion that $Q(u_j,u_{j+1})$ is at least the representative midpoint.

Thus all true eventualities have finite rank in $T_N$.
Finally induct on deletion rounds.
The initial tableau contains $T_N$.
Whenever the current tableau contains $T_N$, all the local witnesses
just exhibited are still available, and monotonicity of the rank
assignment preserves finite ranks.
No type in $T_N$ can therefore be deleted in that round.
\end{proof}

This argument does not infer a uniform finite hitting-time bound
in an infinite-state model.
Its finite object is the set of formula types. An infinite-rank
obstruction would yield an actual infinite counterexample path,
which is enough for the contradiction.

\begin{theorem}[Decidability and complexity]\label{thm:decision}
The tableau procedure accepts $\Lambda$ if and only if $\Lambda$
is satisfiable over $\FVPKS$.
It runs in time $2^{O(n^2)}$ under the size convention of
Section~\ref{sec:prelim}.
Satisfiability for the unbounded PoCTL language over $\FVPKS$
is EXPTIME-complete.
\end{theorem}
\begin{proof}
If $\Lambda$ is satisfiable, a true type containing it survives by
Proposition~\ref{prop:preservation}, so the procedure accepts.
Conversely, if a surviving node contains $\Lambda$, the local checks
and Proposition~\ref{prop:rank-fragments} give a finite
pseudo-possibilistic Hintikka structure for it.
Theorem~\ref{thm:splice} produces a finite PKS model.

For the upper bound, $|S_0|\le2^{|\ecl(\Lambda)|}$ and
$|\ecl(\Lambda)|=O(n^2)$.
The admissible-edge matrix can be computed in polynomial time in
the number of types and the closure size.
For each current tableau, every event rank stabilizes after at most
the number of current nodes many rounds, and each round examines
polynomially many requirements and possible successors.
There are at most $|S_0|$ rounds with actual deletion.
Hence the total time is $2^{O(|\ecl(\Lambda)|)}=2^{O(n^2)}$,
including a polynomial factor for binary rational arithmetic.

For EXPTIME-hardness, use CTL satisfiability~\cite{EH85,Emerson90}.
Translate its existential path quantifier to $\Po_{\ge1}$ and
its universal path quantifier to $\Ne_{>0}$, preserving Boolean
and temporal structure.
A serial CTL model becomes a finitely-valued PKS by assigning
weight $1$ to its edges and $0$ to all other pairs.
Conversely, from a finitely-valued model of a translated formula,
keep precisely the weight-$1$ edges.
This is a serial Kripke structure by attained normalization.
Lemma~\ref{lem:fv} and induction on the formula show that it
satisfies the original CTL formula.
The translation is linear and preserves satisfiability.
Thus the problem is EXPTIME-hard and, by the upper bound, EXPTIME-complete.
\end{proof}

\begin{corollary}[Small-model property]\label{cor:small}
Every satisfiable formula has a finite model of size $2^{O(n^2)}$
whose transition values belong to $EV(\Lambda)$.
\end{corollary}
\begin{proof}
Let $N$ be the number of surviving types and $m$ the number of
eventualities. In Proposition~\ref{prop:rank-fragments}, a fragment
needs at most $N$ interior labels and at most
$N(|\ecl(\Lambda)|+1)$ frontier copies.
There are at most $\max\{1,m\}N$ fragments.
The incoming-weight copies used during splicing multiply the size
by at most $|EV(\Lambda)|$.
A sufficient bound on the number of states is therefore
\begin{equation}\label{eq:model-size}
 \max\{1,m\}\,N^2\bigl(|\ecl(\Lambda)|+2\bigr)|EV(\Lambda)|,
\end{equation}
which is $2^{O(n^2)}$.
All selected edges have weights $r$ or $r^+$ from $EV(\Lambda)$,
and unselected pairs have weight $0$.
\end{proof}

No final identification of all equal-labeled states is needed.
Such an identification could reintroduce cycles or conflate
different chosen edge weights.

\begin{corollary}[Tree-model property]\label{cor:tree}
Every satisfiable formula has a finitely-valued tree model with
branching degree $O(n^2)$.
\end{corollary}
\begin{proof}
In the model selected above, each state retains at most one edge
for each local possibility-next requirement, including normalization.
Its branching degree is therefore $O(|\ecl(\Lambda)|)=O(n^2)$.
Unwind this model from a satisfying state, using finite paths from
that state as tree nodes. A child is obtained by appending one
positive-weight transition, whose weight is retained.
Atomic labels are copied from the last state of the prefix.

The resulting graph is a tree. Paths from a prefix correspond
bijectively to paths from its final state in the finite model,
with the same labels and edge weights.
Induction on formulas, including the quantitative path clauses,
preserves truth. The set of weights is unchanged and finite.
\end{proof}

\begin{example}[A simple accepting tableau]\label{ex:accept}
Let $\Lambda=\Po_{>0.5}(\Diamond a)$.
Then
\[
 V(\Lambda)=\{0,0.5,1\},\qquad
 EV(\Lambda)=\{0,0.25,0.5,0.75,1\},
\]
and its unfolding is
$\Lambda\equiv a\vee\Po_{>0.5}(\Next\Lambda)$.
Consider two states $w_0,w_1$ with $a$ true only at $w_1$ and
$Q(w_0,w_1)=Q(w_1,w_1)=1$.
The unique path from $w_0$ witnesses $\Lambda$.
Their true types $s_0,s_1$ satisfy
\[
 \{\neg a,\Lambda,\Po_{>0.5}(\Next\Lambda)\}\subseteq s_0,
 \qquad \{a,\Lambda\}\subseteq s_1.
\]
Proposition~\ref{prop:actual} gives
$P(s_0,s_1)=P(s_1,s_1)=1$.
Thus $\rank_\Lambda(s_1)=1$ and
$\rank_\Lambda(s_0)=2$.
Their local requirements are witnessed and their true eventualities
have finite ranks, so both types survive.
In constructing a strict-threshold support edge, weight $0.75$
is sufficient for $\Po_{>0.5}(\Next\Lambda)$; the separate
normalization requirement is still witnessed at weight $1$.
\end{example}

\section{The axiomatization for PoCTL}

In this section, we introduce the axiomatization of PoCTL. The axiomatic system proposed in this article is denoted by $\text{AxSys}_{\text{PoCTL}}$. First, we introduce the axiomatic system $\text{AxSys}_{\text{PoCTL}}$, including its axiom schemes and inference rules. Subsequently, we discuss several properties including the deduction theorem of $\text{AxSys}_{\text{PoCTL}}$. Finally, we prove the soundness and { completeness} theorem of the axiomatic system $\text{AxSys}_{\text{PoCTL}}$.

\subsection{The axiomatic system $\text{AxSys}_{\text{PoCTL}}$}

In this section, we give an axiomatization of PoCTL. Its intended semantics throughout this section is $\FVPKS$.
The degenerate formulas in \eqref{eq:endpoints} are treated as
definitional constants, or equivalently as endpoint axioms.
All nontrivial temporal schemes below use nondegenerate bounds.
The axiom schemes of the axiomatic system $\text{AxSys}_{\text{PoCTL}}$ are as follows, where $\sim,\overline{\sim}\in \{\geq, >\}$ and  $\overline{\sim}\not= \sim$:

   \begin{flalign*}
  \text{A1. }& \text{All instances of classical propositional theorems. }&&\\
  \text{A2. }&Po_{\sim r}(\bigcirc (\Phi\vee \Psi))\leftrightarrow (Po_{\sim r}(\bigcirc \Phi)\vee Po_{\sim r}(\bigcirc \Psi))&&\\
  \text{A3a. }&Po_{\sim r}(\lozenge \Phi)\leftrightarrow Po_{\sim r}(true \sqcup \Phi)&&\\
  \text{A3b. }&Ne_{\sim r}(\lozenge \Phi)\leftrightarrow Ne_{\sim r}(true \sqcup \Phi)&&\\
  \text{A4a. }& Po_{\sim r}(\square \Phi)\leftrightarrow \neg Ne_{\overline{\sim} 1-r}(\lozenge \neg \Phi)&&\\
  \text{A4b. }& Ne_{\sim r}(\square \Phi)\leftrightarrow \neg Po_{\overline{\sim} 1-r}(\lozenge \neg \Phi) &&\\
  \text{A5. }&Ne_{\sim r}(\bigcirc \Phi)\leftrightarrow \neg(Po_{\overline{\sim} 1-r}(\bigcirc \neg \Phi))&&\\
  \text{A6. }&[\Psi\vee (\Phi\wedge (Po_{\sim r}(\bigcirc Po_{\sim r}(\Phi\sqcup \Psi)))]\rightarrow Po_{\sim r}(\Phi\sqcup \Psi)&&\\
  \text{A7. }&[\Psi\vee (\Phi\wedge (Ne_{\sim r}(\bigcirc Ne_{\sim r}(\Phi\sqcup \Psi)))]\rightarrow Ne_{\sim r}(\Phi\sqcup \Psi)&&\\
  \text{A8. }&Po_{\geq 1}(\bigcirc true),&&\\
  \text{A9. }&Ne_{\sim r}(\square((\Psi\vee (\Phi\wedge Ne_{\sim r}(\bigcirc \Xi))))\rightarrow \Xi)\rightarrow ( Ne_{\sim r}(\Phi\sqcup \Psi)\rightarrow\Xi)&&\\
  \text{A10. }&Ne_{\sim r}(\square((\Psi\vee (\Phi\wedge Po_{\overline{\sim} 1-r}(\bigcirc \Xi))))\rightarrow \Xi)\rightarrow(Po_{\overline{\sim} 1-r}(\Phi\sqcup \Psi)\rightarrow \Xi)&&\\
  \text{A11. }&Ne_{\sim r}(\square(\Phi\rightarrow \Psi))\rightarrow(Po_{\overline{\sim} 1-r}(\bigcirc \Phi)\rightarrow Po_{\overline{\sim} 1-r}(\bigcirc \Psi))&&\\
  \end{flalign*}

  Axioms for possibilistic measure for any path formula $\varphi$.

  \begin{flalign*}
  \text{AP1. }&Po_{>r}(\varphi)\rightarrow Po_{\geq r}(\varphi)&&\\
    \text{AP2. }&Po_{\geq r_1}(\varphi)\rightarrow Po_{>r_2}(\varphi),\ r_{1}>r_{2}&&\\
    \text{AP3. }&Po_{\geq r_{1}}(\varphi)\rightarrow Po_{\geq r_{2}}(\varphi),\ r_{1}\geq r_{2}&&\\
  \end{flalign*}

  Endpoint instances of AP1--AP3 are read using \eqref{eq:endpoints}.
A8 and these schemes imply every nondegenerate formula
$\Po_{\sim r}(\Next\true)$.
In particular, normalization includes the endpoint
$\Po_{\ge1}(\Next\true)$ itself.

The inference rules are as follows:
\begin{flalign*}
\begin{split}
   &\text{(R1)(MP rule) From }  \Phi\text{ and }\Phi\rightarrow \Psi\text{ infer } \Psi. \\
   & \text{(R2) (Necessitation rules)}\\
    &\text{(1) From a theorem}\  \Phi \text{ infer } Ne_{\sim r}(\bigcirc \Phi);\\
    &\text{(2) from a theorem}\ \Phi \text{ infer } Ne_{\sim r}(\square \Phi).
   \end{split}
\end{flalign*}

Necessitation is applied only to theorems, including in derivations
from assumptions.

Release modalities are definitional duals of the opposite until
modalities:
\[
 \Po_{\sim r}(\Phi\Release\Psi)
 :=\neg\Ne_{\bsim(1-r)}(\neg\Phi\Until\neg\Psi),
 \qquad
 \Ne_{\sim r}(\Phi\Release\Psi)
 :=\neg\Po_{\bsim(1-r)}(\neg\Phi\Until\neg\Psi).
\]
Interval and bounded formulas are treated by the definitional
reductions of Section~\ref{sec:prelim}.
Thus the system axiomatizes the full syntax through its unbounded
nondegenerate core.

\begin{proposition}[Soundness]\label{prop:sound}
Every theorem of $\AxSys$ is valid over $\FVPKS$.
\end{proposition}
\begin{proof}
A1, AP1--AP3, and the duality schemes follow from the semantics.
A2 follows from maxitivity of possibility on a finite union.
A8 follows from attained normalization.
The path splitting in Lemma~\ref{lem:fv} gives A6 and A7.

For A10, suppose its premise and its possibility until antecedent
hold at a state. Choose a finite relevant until witness.
Every state of that prefix satisfies the implication inside the
always modality, since the prefix can be extended by weight-$1$
edges. At the target, $\Psi$ forces $\Xi$.
Working backwards, a successor satisfying $\Xi$ provides the
required possibility-next witness; together with $\Phi$,
the implication forces $\Xi$ at the preceding state.
Thus $\Xi$ holds at the root.

For A9, suppose the premise and its necessity until antecedent
hold, but $\Xi$ fails at the root.
At a current relevant state with $\neg\Xi$, the implication
inside the always modality forces $\neg\Psi$.
The until antecedent then supplies $\Phi$, so the relevant
necessity-next formula for $\Xi$ must fail.
Choose a relevant successor with $\neg\Xi$.
The until formula propagates to this successor.
Iteration produces an infinite relevant path avoiding $\Psi$,
contradicting the necessity until antecedent by
Lemma~\ref{lem:fv}. This also proves the strict-path case,
where finite values keep the path possibility strictly above
its threshold.

For A11, the always premise ensures $\Phi\to\Psi$ at the successor
of any relevant next-step witness.
Rules R1 and R2 preserve validity. The definitional reductions
were proved semantically in Section~\ref{sec:prelim}.
\end{proof}

The domain restrictions are substantive. An A9 instance with
$\Ne_{\ge0}$ and $\Xi=\false$ would be invalid, and a rule deriving
$\Ne_{>1}(\Next\true)$ from $\true$ would also be invalid.
These instances are excluded before any unfolding or necessitation.

\subsection{Derived rules and the deduction theorem}

\begin{lemma}[Monotonicity, unfoldings, and replacement]\label{lem:replacement}
The following are derivable in $\AxSys$:
\begin{enumerate}[label=(\roman*),leftmargin=2.3em]
\item If $\vdash\Phi\to\Psi$, then
$\vdash\Po_{\sim r}(\Next\Phi)\to\Po_{\sim r}(\Next\Psi)$
and
$\vdash\Ne_{\sim r}(\Next\Phi)\to\Ne_{\sim r}(\Next\Psi)$.
\item All equivalences represented by the conjunctive and disjunctive
      decompositions in parts (a)--(b) of Table~\ref{tab:main_table}.
\item Provable equivalents can be replaced in every state-formula
context of the language.
\item If $\Phi'$ is obtained from $\Phi$ by the positive-normal-form
procedure of Section~\ref{sec:prelim}, then
$\vdash\Phi\leftrightarrow\Phi'$.
For bounded formulas this statement uses their specified definitional
expansions. In particular, the members of each complementary pair
in the closure are provably negations of each other.
\end{enumerate}
\end{lemma}
\begin{proof}
For possibility-next monotonicity, necessitate the theorem
$\Phi\to\Psi$ with the complementary necessity-always bound,
and apply A11 and R1.
Apply this result to $\neg\Psi\to\neg\Phi$ and use A5
to obtain necessity-next monotonicity.

For the converse direction of an until unfolding, first let
$E=\Po_{\sim r}(\Phi\Until\Psi)$ and abbreviate
\[
 \Xi=\Psi\vee(\Phi\wedge\Po_{\sim r}(\Next E)).
\]
A6 gives $\vdash\Xi\to E$.
Part (i) gives
$\vdash\Po_{\sim r}(\Next\Xi)\to\Po_{\sim r}(\Next E)$.
Consequently,
\[
 \vdash[\Psi\vee(\Phi\wedge\Po_{\sim r}(\Next\Xi))]\to\Xi.
\]
Necessitate this theorem and apply the appropriate instance of A10.
It follows that $\vdash E\to\Xi$.
Together with A6 this yields the full unfolding.
For necessity until, the proof uses A7 and A9 instead.
Release unfoldings follow by the definitional dualities.
This argument does not presuppose general replacement.

Replacement inside next formulas follows from (i) in both
directions. For until, suppose
$\vdash\Phi\leftrightarrow\Phi'$ and
$\vdash\Psi\leftrightarrow\Psi'$.
Use the primed until formula as $\Xi$ in the induction scheme
for the unprimed formula.
Its A6 or A7 instance, together with the two propositional
equivalences, gives the required theorem premise.
A10 or A9 therefore derives the unprimed-to-primed implication.
Interchanging the formulas gives the converse.
The defined operators and interval reductions preserve this property.
Induction on the surrounding context proves (iii).
Every necessitation used here is applied to a theorem.

For (iv), Boolean negation is moved inward using A1.
The next-step dualities follow from A5, and the until/release
dualities follow from the displayed definitions of release.
The eventually/always conversions use A3 and A4.
Interval splitting and endpoint simplification are the specified
abbreviations and endpoint axioms; bounded operators are expanded
by their finite definitional recursions.
Part (iii) permits each of these conversions inside its surrounding
state-formula context. Induction on the finite sequence of conversions
proves $\vdash\Phi\leftrightarrow\Phi'$.
Applying this to the normal form of $\neg\Phi$ proves the final claim.
This is a syntactic consequence of the axioms and definitions;
it does not use the completeness theorem.
\end{proof}

\begin{theorem}[Deduction theorem]\label{thm:deduction}
For every set $\Gamma$ of state formulas and all $\Phi,\Psi$,
\[
 \Gamma\cup\{\Phi\}\vdash\Psi
 \quad\Longleftrightarrow\quad
 \Gamma\vdash\Phi\to\Psi.
\]
\end{theorem}
\begin{proof}
The right-to-left implication follows by R1.
For the converse, induct on the derivation from
$\Gamma\cup\{\Phi\}$ using the usual propositional argument.
A necessitation step can only have a theorem as its premise,
so its conclusion is itself a theorem and is handled like an axiom.
Thus no necessitation of a formula depending on $\Phi$ is needed.
See also the standard modal treatment in~\cite{r39}.
\end{proof}

\begin{lemma}[Derived refutation principles]\label{lem:traps}
At nondegenerate bounds, the following are derivable.
\begin{enumerate}[label=(\arabic*),leftmargin=2.3em]
\item
$\bigl[\Po_{\ge r}(\Next\Phi)\wedge
\Ne_{>1-r}(\Next\Psi)\bigr]\to
\Po_{\ge r}(\Next(\Phi\wedge\Psi))$.
The corresponding strict formula uses
$\Po_{>r}$ and $\Ne_{\ge1-r}$.
\item If
$\vdash\Xi\to[\neg\Psi\wedge
(\Phi\to\Ne_{>1-r}(\Next\Xi))]$, then
$\vdash\Xi\to\neg\Po_{\ge r}(\Phi\Until\Psi)$.
\item If
$\vdash\Xi\to[\neg\Psi\wedge
(\Phi\to\Ne_{\ge1-r}(\Next\Xi))]$, then
$\vdash\Xi\to\neg\Po_{>r}(\Phi\Until\Psi)$.
\item If
$\vdash\Xi\to[\neg\Psi\wedge
\Po_{\ge1-r}(\Next\Xi)]$, then
$\vdash\Xi\to\neg\Ne_{>r}(\Phi\Until\Psi)$.
\item If
$\vdash\Xi\to[\neg\Psi\wedge
\Po_{>1-r}(\Next\Xi)]$, then
$\vdash\Xi\to\neg\Ne_{\ge r}(\Phi\Until\Psi)$.
\end{enumerate}
\end{lemma}
\begin{proof}
For (1), propositional logic gives
$\Phi\to[(\Phi\wedge\Psi)\vee\neg\Psi]$.
Next monotonicity and A2 yield
\[
 \vdash\Po_{\ge r}(\Next\Phi)\to
 [\Po_{\ge r}(\Next(\Phi\wedge\Psi))
       \vee\Po_{\ge r}(\Next\neg\Psi)].
\]
Conjoin $\Ne_{>1-r}(\Next\Psi)$ and use A5 to exclude the
second disjunct. The strict case is identical.

For (2), A5 turns its premise into a theorem implying
\[
 \vdash
 [\Psi\vee(\Phi\wedge\Po_{\ge r}(\Next\neg\Xi))]
 \to\neg\Xi.
\]
Necessitate this theorem and apply A10 with $\neg\Xi$ as
its induction formula. The result is
$\vdash\Po_{\ge r}(\Phi\Until\Psi)\to\neg\Xi$.
Its contrapositive proves (2).
The proof of (3) replaces the next possibility bound by $>r$
and its dual necessity bound by $\ge1-r$.

For (4), its premise makes $\Xi$ incompatible with both $\Psi$
and $\Ne_{>r}(\Next\neg\Xi)$.
Hence
\[
 \vdash
 [\Psi\vee(\Phi\wedge\Ne_{>r}(\Next\neg\Xi))]
 \to\neg\Xi.
\]
Necessitation and A9 give
$\vdash\Ne_{>r}(\Phi\Until\Psi)\to\neg\Xi$.
Taking the contrapositive proves (4).
For (5), replace the necessity bound by $\ge r$ and its
dual possibility bound by $>1-r$.
\end{proof}

Taking $\Phi=\true$ in parts (2)--(5) yields the corresponding
principles for eventually formulas.

\subsection{Complete types and the syntactic role of \texorpdfstring{$D(s,t)$}{D(s,t)}}
\subsubsection{Complete-type decomposition}

Use the original characteristic-formula notation
\[
 \wedge s=\wedge_{\Phi\in s}\Phi.
\]
No additional notation for this conjunction is needed.

\begin{lemma}[Complete-type decomposition]\label{lem:partition}
For the initial type set $S_0$,
\begin{equation}\label{eq:partition}
 \vdash\true\leftrightarrow\vee_{s\in S_0}\wedge s,
 \qquad
 \vdash\Phi\leftrightarrow
        \vee_{\substack{s\in S_0\\\Phi\in s}}\wedge s
        \quad(\Phi\in\ecl(\Lambda)).
\end{equation}
Different complete types have pairwise incompatible characteristic
formulas.
\end{lemma}
\begin{proof}
Choose one representative from each complementary pair in the
closure. By Lemma~\ref{lem:replacement}(iv), each pair consists
of formulas provably equivalent to a formula and its negation.
The conjunction of the corresponding excluded-middle
theorems, distributed finitely, is the disjunction of all complete
selections.

Any complete selection that is not an initial type violates PCR,
the truth-constant conditions, or normalization. Each of these
requirements is provable: the Boolean and truth-constant conditions
use A1; the temporal decompositions use
Lemma~\ref{lem:replacement}; the threshold decompositions in (PC1)
use AP1--AP3, A5, and propositional logic; and normalization uses A8.
Thus the characteristic formula of every excluded selection is
refutable. Removing these finitely many disjuncts proves the first
equivalence in \eqref{eq:partition}.

Conjoin $\Phi$ with that equivalence and distribute.
If $\Phi\in s$, the resulting conjunct is equivalent to
$\wedge s$; otherwise the complement of $\Phi$ belongs to $s$,
so the conjunct is false. This proves the second equivalence.
Two distinct complete types differ on a complementary pair,
which also proves pairwise incompatibility.
\end{proof}

Suppose an induction has already refuted every node deleted before
a current stage $S_i$. Those disjuncts can then be removed from
\eqref{eq:partition}, giving
\begin{equation}\label{eq:current-partition}
 \vdash\true\leftrightarrow\vee_{t\in S_i}\wedge t,
 \qquad
 \vdash\Phi\leftrightarrow
        \vee_{\substack{t\in S_i\\\Phi\in t}}\wedge t.
\end{equation}
This is the precise meaning of omitting previously deleted types
by the induction hypothesis. It does not presuppose that all
currently remaining types are consistent.

\subsubsection{Excluding inadmissible successor types}

\begin{lemma}[Edge exclusion]\label{lem:edge-exclusion}
For input thresholds $r\in V(\Lambda)$ at nondegenerate bounds,
\begin{align}
 P(s,t)<r
 &\ \Longrightarrow\
 \vdash\wedge s\to
       \neg\Po_{\ge r}(\Next\wedge t),\label{eq:exclude-ge}\\
 P(s,t)\le r
 &\ \Longrightarrow\
 \vdash\wedge s\to
       \neg\Po_{>r}(\Next\wedge t).\label{eq:exclude-gt}
\end{align}
\end{lemma}
\begin{proof}
If $P(s,t)<r$, then $r\notin D(s,t)$.
The first part of \eqref{eq:D} therefore supplies a formula
$\Ne_{>1-r}(\Next\Phi)\in s$ with $\Phi\notin t$.
Completeness gives $\neg\Phi\in t$. Thus
\[
 \vdash\wedge s\to\Ne_{>1-r}(\Next\Phi),
 \qquad
 \vdash\wedge t\to\neg\Phi.
\]
By next monotonicity,
\[
 \vdash\Po_{\ge r}(\Next\wedge t)
       \to\Po_{\ge r}(\Next\neg\Phi).
\]
A5 makes this consequent incompatible with the necessity-next
formula forced by $\wedge s$. This proves \eqref{eq:exclude-ge}.

If $P(s,t)\le r$, then $r^+\notin D(s,t)$.
Because $r^+=1-(1-r)^-$, the second part of $D$ supplies
$\Ne_{\ge1-r}(\Next\Phi)\in s$ with $\neg\Phi\in t$.
Next monotonicity gives
$\Po_{>r}(\Next\wedge t)\to\Po_{>r}(\Next\neg\Phi)$.
Applying A5 proves \eqref{eq:exclude-gt}.
\end{proof}

The threshold in these two claims belongs to $V(\Lambda)$.
The midpoint $r^+$ is used to inspect the second part of $D$;
it is not substituted for the strict formula $\Po_{>r}$
inside an axiomatic derivation.

\begin{lemma}[Local successor decomposition]\label{lem:support}
Assume the current decomposition \eqref{eq:current-partition}.
If $\Po_{\ge r'}(\Next\Xi)\in s$, then
\begin{equation}\label{eq:support-ge}
 \vdash\wedge s\to
 \Po_{\ge r'}\left(\Next
  \vee_{\substack{t\in S_i,\ \Xi\in t\\P(s,t)\ge r'}}
  \wedge t\right).
\end{equation}
If $\Po_{>r'}(\Next\Xi)\in s$, then
\begin{equation}\label{eq:support-gt}
 \vdash\wedge s\to
 \Po_{>r'}\left(\Next
  \vee_{\substack{t\in S_i,\ \Xi\in t\\P(s,t)>r'}}
  \wedge t\right).
\end{equation}
\end{lemma}
\begin{proof}
For the first formula, \eqref{eq:current-partition}, replacement,
and repeated applications of A2 give
\[
 \vdash\Po_{\ge r'}(\Next\Xi)
 \leftrightarrow
 \vee_{\substack{t\in S_i\\\Xi\in t}}
       \Po_{\ge r'}(\Next\wedge t).
\]
Under $\wedge s$, each disjunct with $P(s,t)<r'$
is excluded by Lemma~\ref{lem:edge-exclusion}.
Applying A2 to the remaining finite disjunction yields
\eqref{eq:support-ge}.
The proof of \eqref{eq:support-gt} instead excludes
$P(s,t)\le r'$ using \eqref{eq:exclude-gt}.
\end{proof}

An empty disjunction denotes $\false$.
At every nondegenerate bound,
\[
 \vdash\neg\Po_{\sim r}(\Next\false),
\]
because necessitation applies to the theorem $\true$ and A5 gives
the corresponding dual formula.
The local decomposition also covers A8 with $r'=1$ and $\Xi=\true$.

\subsection{Weak completeness}\label{sec:completeness}

\begin{theorem}[Completeness over finitely-valued PKSs]\label{thm:complete}
The system $\AxSys$ is sound and weakly complete over $\FVPKS$:
\[
 \models_{\mathrm{FV}}\Lambda
 \quad\Longleftrightarrow\quad
 \vdash\Lambda.
\]
\end{theorem}
\begin{proof}
Soundness is Proposition~\ref{prop:sound}.
Lemma~\ref{lem:replacement}(iv) allows every positive-normal-form
conversion used in the tableau to be made by a provable equivalence.
Thus the following use of normalized complements is valid both
semantically and within the proof system.
For completeness, we first prove the deletion-refutation claim:
\begin{equation}\label{eq:delete-refutation}
 s\text{ is deleted}\quad\Longrightarrow\quad
 \vdash\neg\wedge s.
\end{equation}
Induct on the deletion round.
The induction hypothesis yields the current decomposition
\eqref{eq:current-partition}.
All nodes deleted in the same round use this common current set;
their refutations do not assume that other nodes in that round
have already been refuted.

\medskip
\noindent\emph{Case 1: failure of normalization.}
If $s$ has no weight-$1$ allowed successor,
Lemma~\ref{lem:support} applied to A8 yields
$\vdash\wedge s\to\Po_{\ge1}(\Next\false)$.
The latter is refutable, so $\vdash\neg\wedge s$.

\medskip
\noindent\emph{Case 2: a missing local possibility-next witness.}
Suppose $\Po_{\sim r'}(\Next\Xi)\in s$ has no allowed successor
at its own bound. The relevant formula in
Lemma~\ref{lem:support} has an empty disjunction, so again
$\wedge s$ implies a refutable
$\Po_{\sim r'}(\Next\false)$.
This covers both strict and non-strict local requirements.

\medskip
\noindent\emph{Case 3: a possibility until eventuality has infinite rank.}
First let $E=\Po_{\ge r}(\Phi\Until\Psi)\in s$ with
$\rank_E(s)=\infty$.
Keep the original infinite-rank set
\begin{equation}\label{eq:bad-po}
 T_E=\{t\in S_i:\rank_E(t)=\infty\},
 \qquad
 \Xi_E=\vee_{t\in T_E}\wedge t.
\end{equation}
There is no additional restriction $E\in t$ in this case.

Every $t\in T_E$ lacks $\Psi$, since a target has rank $1$.
Therefore
\[
 \vdash\Xi_E\to\neg\Psi.
\]
Moreover, if $t\in T_E$, $\Phi\in t$, and $u\in S_i\setminus T_E$,
then $P(t,u)<r$.
Otherwise the finite rank of $u$ and the min rule
\eqref{eq:rank-po} would assign $t$ a finite rank.

Complete-type decomposition and mutual incompatibility give
\begin{equation}\label{eq:bad-complement}
 \vdash\neg\Xi_E
 \leftrightarrow
 \vee_{u\in S_i\setminus T_E}\wedge u.
\end{equation}
For each such $u$, \eqref{eq:exclude-ge} yields
\[
 \vdash\wedge t\to
       \neg\Po_{\ge r}(\Next\wedge u).
\]
Conjoin these finitely many implications, use A2, and replace
the disjunction in \eqref{eq:bad-complement}. It follows that
\[
 \vdash\wedge t\to
       \neg\Po_{\ge r}(\Next\neg\Xi_E),
 \qquad
 \vdash\wedge t\to\Ne_{>1-r}(\Next\Xi_E).
\]
The second conclusion uses A5.

If $\Phi\notin t$, completeness gives $\neg\Phi\in t$, so
$\wedge t\to(\Phi\to\Ne_{>1-r}(\Next\Xi_E))$
holds propositionally.
Consequently, disjoining over all $t\in T_E$ yields
\begin{equation}\label{eq:case3}
 \vdash\Xi_E\to
 [\neg\Psi\wedge(\Phi\to\Ne_{>1-r}(\Next\Xi_E))].
\end{equation}
Lemma~\ref{lem:traps}(2), derived from A10, gives
$\vdash\Xi_E\to\neg E$.
The deleted node $s$ belongs to $T_E$ and contains $E$.
Thus $\wedge s$ entails both $E$ and $\neg E$,
and $\vdash\neg\wedge s$.

If $E=\Po_{>r}(\Phi\Until\Psi)$, use the same definition
\eqref{eq:bad-po}. Infinite-rank nodes containing $\Phi$
have no edge of weight greater than $r$ to a finite-rank node.
Thus $P(t,u)\le r$ in the preceding argument.
Use \eqref{eq:exclude-gt}, A2, and A5 to obtain
\[
 \vdash\Xi_E\to
 [\neg\Psi\wedge(\Phi\to\Ne_{\ge1-r}(\Next\Xi_E))].
\]
Lemma~\ref{lem:traps}(3) gives $\Xi_E\to\neg E$,
and the same contradiction follows.

\medskip
\noindent\emph{Case 4: a possibility eventually formula has infinite rank.}
This is Case 3 with $\Phi=\true$.

\medskip
\noindent\emph{Case 5: a necessity until eventuality has infinite rank.}
First let $E=\Ne_{>r}(\Phi\Until\Psi)\in s$ and
$\rank_E(s)=\infty$.
Here use
\begin{equation}\label{eq:bad-ne}
 T_E=\{t\in S_i:E\in t,\ \rank_E(t)=\infty\},
 \qquad
 \Xi_E=\vee_{t\in T_E}\wedge t.
\end{equation}
In contrast to Case 3, the explicit $E$-membership restriction
is appropriate for this max--min argument.

For every $t\in T_E$, the target is absent and the unfolding gives
\[
 \neg\Psi\in t,\qquad \Phi\in t,\qquad
 \Ne_{>r}(\Next E)\in t.
\]
Hence $\vdash\Xi_E\to\neg\Psi$.

Fix $t\in T_E$. If every relevant local requirement had a
finite-$E$-rank allowed successor, the finite maximum of the
corresponding finite minima would make $\rank_E(t)$ finite.
Thus some requirement in Table~\ref{tab:relevant} has only
infinite-$E$-rank allowed successors.

Suppose first that this requirement is
$\Po_{\ge r'}(\Next\Xi)\in t$.
Then $r'\ge1-r$. Each allowed successor $u$ satisfies
$P(t,u)\ge r'\ge1-r$.
Condition (LC0), applied to $\Ne_{>r}(\Next E)\in t$,
therefore gives $E\in u$. All those successors lie in $T_E$.
Lemma~\ref{lem:support} and next monotonicity yield
\[
 \vdash\wedge t\to\Po_{\ge r'}(\Next\Xi_E),
\]
and AP3 yields
\[
 \vdash\wedge t\to\Po_{\ge1-r}(\Next\Xi_E).
\]
If the selected requirement is $\Po_{>r'}(\Next\Xi)$,
then again $r'\ge1-r$, and the same successor-membership argument
applies. Formula \eqref{eq:support-gt} first gives
$\wedge t\to\Po_{>r'}(\Next\Xi_E)$.
AP1 and AP3 give the same lower-threshold conclusion.
The normalization requirement is the first case with
$r'=1$ and $\Xi=\true$.

Disjoining these conclusions over $T_E$ gives
\begin{equation}\label{eq:case5}
 \vdash\Xi_E\to
 [\neg\Psi\wedge\Po_{\ge1-r}(\Next\Xi_E)].
\end{equation}
Lemma~\ref{lem:traps}(4), derived from A9, now gives
$\vdash\Xi_E\to\neg E$.
Since $s\in T_E$ and $E\in s$, this refutes $\wedge s$.

For $E=\Ne_{\ge r}(\Phi\Until\Psi)$, keep
\eqref{eq:bad-ne}. The relevant non-strict requirements have
$r'>1-r$, so their decompositions imply
$\Po_{>1-r}(\Next\Xi_E)$ by AP2.
The relevant strict requirements have $r'\ge1-r$.
At equality they already give the required formula;
at a strictly larger threshold, AP1 and AP2 give it.
Every allowed successor has weight greater than $1-r$,
so $\Ne_{\ge r}(\Next E)\in t$ again ensures that it contains $E$.
Thus
\[
 \vdash\Xi_E\to
 [\neg\Psi\wedge\Po_{>1-r}(\Next\Xi_E)].
\]
Apply Lemma~\ref{lem:traps}(5) to derive $\Xi_E\to\neg E$
and refute $\wedge s$.

\medskip
\noindent\emph{Case 6: a necessity eventually formula has infinite rank.}
This is Case 5 with $\Phi=\true$.

\medskip
The cases establish \eqref{eq:delete-refutation} by induction.
Now suppose $\models_{\mathrm{FV}}\Lambda$.
Then $\neg\Lambda$ is unsatisfiable over $\FVPKS$.
By Theorem~\ref{thm:decision}, every initial type containing
$\neg\Lambda$ is deleted in the tableau for $\neg\Lambda$.
The deletion claim refutes all their characteristic formulas.
Apply \eqref{eq:partition} to $\neg\Lambda$ to derive
$\vdash\neg\neg\Lambda$, and hence $\vdash\Lambda$.
When $\neg\Lambda$ was normalized before the tableau was formed,
Lemma~\ref{lem:replacement}(iv) transfers this refutation back
to the original $\neg\Lambda$.
\end{proof}

The two definitions of $T_E$ serve different purposes.
For possibility until, all infinite-rank types form a set closed
under relevant successors whenever $\Phi$ holds.
For necessity until, one selected relevant requirement has all
its allowed successors in the $E$-labeled infinite-rank set.
This is why the original characteristic-formula argument works
with independent min and max--min ranks.

The formula of Example~\ref{ex:no-fmp} can also be refuted directly,
without invoking completeness. Its definition gives the propositional
theorem
\[
 \vdash\neg a\to
 [\neg a\wedge(\chi\to\Ne_{>1/2}(\Next\neg a))].
\]
Lemma~\ref{lem:traps}(2), with $\Xi=\neg a$, $\Phi=\chi$,
$\Psi=a$, and $r=1/2$, yields
\[
 \vdash\neg a\to\neg\Po_{\ge1/2}(\chi\Until a).
\]
Since $\vdash\chi\to\neg a$, it follows that
$\vdash\neg\Lambda_\infty$.
This derivation explicitly matches the finitely-valued semantic
restriction: Example~\ref{ex:no-fmp} satisfies $\Lambda_\infty$
in an unrestricted countable PKS.

\begin{corollary}[Finite-premise completeness]\label{cor:finite-premises}
For every finite set $\Gamma$ of PoCTL formulas and formula $\Phi$,
\[
 \Gamma\models_{\mathrm{FV}}\Phi
 \quad\Longleftrightarrow\quad
 \Gamma\vdash\Phi.
\]
\end{corollary}
\begin{proof}
Apply Theorem~\ref{thm:complete} to
$(\wedge\Gamma)\to\Phi$ and use
Theorem~\ref{thm:deduction}. The reverse implication follows
from soundness.
\end{proof}

\begin{remark}[Noncompactness]\label{re:compact of PoCTL}
Since CTL is a proper sublogic of PoCTL (\cite{PoCTL2015}), it follows, like CTL, PoCTL is not compact, i.e., PoCTL containing infinite sets of formulae which are not satisfiable but every finite subset of them is satisfiable. In fact, we can give such instance using pure PoCTL formulae instead of CTL formulaea as follow. According to the semantics of PoCTL, the set
  \begin{equation*}
    \Gamma=\{Po_{<1}(\bigcirc a)\}\cup\{Po_{\geq 1-1/n}(\bigcirc a)\ |\ n\in\mathbb{N}^{+},n\geq 1\}
  \end{equation*}
is an unsatisfiable set of formulas, where $a\in AP$. It is evident that every finite subset of $\Gamma$ is satisfiable. Hence, this logic is not compact. In this situation, any finitary axiomatization fails to be strongly complete \cite{infinite-axiom-1992}. To achieve strong completeness, we need employ some infinitary rules as done in \cite{He25}. We will leave the strong complete axiomatization of PoCTL in the next study.
\end{remark}

\section{Conclusion}\label{sec:conclusion}

We have established satisfiability and axiomatization results for
PoCTL over normalized finitely-valued possibilistic Kripke structures.
The restriction allows countably infinite state spaces, while ensuring
the attainment of state-local quantitative path witnesses.
A counterexample separates this semantics from unrestricted
countable PKSs and shows that the latter class does not have the
finite-model property.

The decision procedure extracts the finite threshold scale
$EV(\Lambda)$, uses the original two-part set $D(s,t)$ to determine
admissible transition weights, and computes eventuality ranks
independently. A possibility rank supplies a decreasing witness
path. A necessity rank supplies a decreasing successor for every
relevant local requirement at that requirement's own threshold.
Finite supported fragments and cyclic splicing then produce a
finite model. The resulting procedure has a $2^{O(n^2)}$ upper
bound and gives EXPTIME-completeness for the unbounded language,
as well as small-model and bounded-branching tree-model properties.

The finitary system $\AxSys$ is sound and weakly complete for the
same class. Its proof uses complete-type decomposition, the
syntactic edge-exclusion consequences of $D$, and refutations
corresponding to each deletion rule. In particular, the two
eventuality refutations follow from the existing induction schemes.
Normalization at weight $1$, nondegenerate threshold domains,
and the separation of strict from non-strict successor requirements
are essential throughout.

Several questions remain beyond these results.
A procedure for unrestricted countable PKSs must account for
unattained suprema and limiting path values; the counterexample
excludes a direct finite-model argument but does not by itself
settle decidability by other methods.
For bounded operators, the complexity of a succinct binary horizon
must be distinguished from the size after explicit unfolding.
Strong completeness also requires a treatment of infinite premise
sets, as the noncompactness example demonstrates.

Of course, we still have many issues to solve. PoCTL cannot completely solve the formal description and corresponding verification of real systems in fuzzy environments, so we propose a generalized possibilistic computation tree logic (GPoCTL) (\cite{GPoCTL2015}) or more extensive generalized possibilistic fuzzy temporal logic (GPoFTL) (\cite{GPoFLTL}) (including fuzzy temporal operators such as“soon”,  “presently”, “gradually”, “within”, “last”, “nearly always”, “almost always”, “in the long distant future”, “in the middle of”,  “nearly until” , “almost until”) and the generalized possibility computation tree logic with frequency (GPoCTL$_F$) in \cite{He24}. Under these logics, the satisfiability problem and axiomatic problem are more complex issues, and in some case undecidable, c.f.\cite{Frigeri14,Almagor16,Vidal22}. As mentioned in \cite{Fagina24,Aguilera25}, the problem of axiomatization under fuzzy logic and fuzzy temporal logic is a more challenging one, the satisfiability and axiomatization problems of GPoCTL or GPoFTL require more quantitative techniques to handle, and we will leave them as the future topic to be studied.

In addition, we have recently established a generalized possibilistic decision-making process (GPDP) for practical systems with nondeterministic action selection \cite{li23,liu23}. How to choose the optimal scheduler or policy to attain the object (described by some temporal formulae) is a challenging task of GPDP. As mentioned in  \cite{Chatterjee10}, many problems in probabilistic decision-making process are undecidable.  Therefore, it is interesting to study whether the related satisfiability problems of GPoCTL or GPoFTL under GPDP are still decidable, and provide relevant algorithms. This forms another future topic to be studied.

\section*{Acknowledgments}

The authors would like to thank the anonymous referees for helping them refine the ideas presented in this paper and improve the clarity of the presentation.
%The authors would also like to express their special thanks to Dr. Licong Cui at Case Western Reserve University for detailed suggestions that improved the paper's quality.
This article was supported by the National Science Foundation of China (Grants No. 12471437, 12071271) and the Shaanxi Fundamental Science Research
Project for Mathematics and Physics (Grant No. 23JSZ011).

\end{document}